\documentclass[journal ]{new-aiaa}
\usepackage[utf8]{inputenc}
\usepackage{textcomp}

\usepackage{graphicx}
\usepackage{amsmath}

\usepackage{amsthm}
\usepackage[version=4]{mhchem}
\usepackage{siunitx}
\usepackage{longtable,tabularx}
\usepackage{comment}                              
\usepackage[super]{nth}

\usepackage{bm}

\newtheorem{definition}{Definition}[section]
\newtheorem{remark}{Remark}[section]

\newtheorem{assumption}{Assumption}[section]
\newtheorem{proposition}{Proposition}[section]

\usepackage{hyperref}

\title{Nuisance-free Ground Collision Avoidance Design with Adaptive Exponential Control Barrier Functions}

\author{Ege Ç. Altunkaya \footnote{Research Assistant, Aviation Institute, Aerospace Research Center, Istanbul Technical University, Istanbul, Türkiye, 34469; altunkaya16@itu.edu.tr (Corresponding Author)} and İbrahim Özkol \footnote{Professor, Aerospace Research Center, Istanbul Technical University, Istanbul, Türkiye, 34469.}}
\affil{Aerospace Research Center, Istanbul Technical University, Istanbul, Türkiye, 34469.}

\begin{document}

\maketitle

\begin{abstract}
    The significance of the automatic ground collision avoidance system (Auto-GCAS) has been proven by considering the fatal crashes that have occurred over decades. Even though extensive efforts have been put forth to address the ground collision avoidance in the literature, the notion of being nuisance-free has not been sufficiently addressed. In this study, the Auto-GCAS design is formulated by merging exponential control barrier functions with sliding manifolds to manipulate the barrier function dynamics. The adaptive properties of the sliding manifolds are tailored to the key and governing flight parameters, ensuring that the nuisance-free requirement is satisfied. Furthermore, to ensure all safety requirements are met, a flight envelope protection algorithm is designed using control barrier functions to assess the commands generated by the Auto-GCAS. Eventually, the performance of the proposed methodology is demonstrated, focusing on authority-sharing, collision avoidance capability, and nuisance-free operation through various scenarios and Monte Carlo simulations. Simulation results demonstrate that the proposed adaptive exponential CBF-based Auto-GCAS achieves a $99.88\%$ ground collision avoidance success rate across diverse scenarios, without nuisance activations and while respecting aircraft dynamic limits.
\end{abstract}

\section{Introduction}
\lettrine{A}{s} one of the major causes of fatal aircraft accidents, controlled flight into terrain (CFIT) refers to accidents resulting from in-flight collisions with terrain, water, or obstacles, without any indication of a loss of control \cite{safety1, safety2}. The fact that the aircraft is still being controlled by the crew at the time of collision in this type of accident reveals that human error is the most probable cause \cite{safety1}. Even though the ground proximity warning systems (GPWS) are designed to support the pilot in taking action to prevent crashes\textemdash and have resulted in a remarkable mitigation in the number of accidents \cite{safety3}\textemdash seemingly, they do not appear to be a conclusive solution for avoiding fatal crashes \cite{AgcasHistory2} since they rely on pilot intervention. On the other hand, gravity-induced loss of consciousness (G-LOC) during highly complex military operations is another cause of CFIT \cite{safety4, AgcasHistory, safety6}. According to the research in \cite{safety5}, the total duration of a pilot's incapacitation following G-LOC is 28 seconds, an astonishing duration that inherently increases the risk of CFIT \cite{safety4, AgcasHistory}. Therefore, in the event of a pilot's blackout, pilot support systems such as GPWS or enhanced GPWS (EGPWS) lose their effectiveness, as they are warning systems, requiring manual response, rather than automation systems. At this point, in order to reduce the number of mishap of CFIT, the technology of Automatic Ground Collision Avoidance System (Auto-GCAS) emerged, generating automated recovery maneuvers to prevent ground collision. This technology has been developed over the past three decades in the U.S. and was implemented in the F-16 in November 2014 \cite{safety7, AgcasHistory2}. From its implementation until 2019, it has saved eight pilots and seven aircraft, according to the statistics in \cite{AgcasHistory2}; furthermore, it is projected to save 11 F-16s, nine pilots, and \$400 million over the remaining operational life of the F-16 \cite{safety7}.

Although the development of Auto-GCAS technology has provided significant advantages in mitigating CFIT risks, it also presents challenges stemming from human-autonomy interaction \cite{runtimeMagazine}. In this context, the unnecessary, untimely, and interruptive interventions of the Auto-GCAS are terminologically referred to as ``nuisance" \cite{runtimeMagazine}. These interventions may include premature activation or false alarms during maneuvers, disrupting pilot operations and trust in the system. Moreover, interventions that are impulse-like or insufficiently aggressive may lead the pilot to perceive them as unnecessary, as they could have executed the required recovery maneuver using a more aggressive and abrupt approach \cite{Agcas3, Agcas11}. Therefore, to ensure harmony between human and autonomy in piloted aircraft, the design of a nuisance-free Auto-GCAS is paramount, that is why, a number of studies have been conducted in order to address this issue. Nevertheless, primarily, the qualitative and quantitative definitions of ``nuisance-free" must be established. Fortunately, there is consensus on the qualitative definition: a nuisance-free intervention is characterized as ``timely and aggressive" \cite{runtimeMagazine, Agcas11, Agcas7, Agcas6, Agcas3}. This implies that the Auto-GCAS algorithm must generate a recovery maneuver that satisfies the following criteria: \textbf{(1)} the applied command must have maximum allowable magnitude, such as maximum pitch rate or maximum load factor, depending on the flight control architecture, and \textbf{(2)} the aircraft’s closest point to the ground must match the pre-defined keep-out zone, i.e. buffer. In the quantitative sense, Eq.~\eqref{nuisanceFreeQuantification} describes the notion of nuisance-free, as clearly presented in \cite{Agcas11}.

\begin{equation}
    \label{nuisanceFreeQuantification}
        \begin{split}
            \min_{t\in[t_0,t_f]} ||h_\text{A/C}(t, \bm{u}(t)) - h_\text{DTED}(t)||_2 = h_\text{buff} \\
            u_{i}(t) = 
                \begin{cases}
                    u_{\text{max}_i}, \forall t \in [t_0,t_\text{CPA}) \\
                    \text{or} \\
                    u_{\text{min}_i}, \forall t \in [t_0,t_\text{CPA}) \\
                \end{cases} 
        \end{split}
\end{equation}
where $\bm{u}(t):\mathbb{R} \to \mathbb{R}^\text{m}$ is the control input vector. Additionally, $h_\text{A/C}(t,\bm{u}(t)):\mathbb{R} \times \mathbb{R}^\text{m} \to \mathbb{R}$, $h_\text{DTED}(t): \mathbb{R} \to \mathbb{R}$, and $h_\text{buff} \in \mathbb{R}_{>0}$ are the altitude of aircraft, height of digital terrain elevation data (DTED), and height of pre-defined scalar buffer, respectively. Understandably, these two constraints imply that, for a given control input $\bm{u}(t)$, a timely aircraft trajectory corresponds to a recovery maneuver $h_\text{A/C}(t, \bm{u}(t))$ that approaches but does not violate the buffer within the time interval $t\in[t_0 \hspace{0.15cm} t_f]$, also it refers to a ``last second" maneuver \cite{Agcas11}. Furthermore, at least one control input $u_i(t)$ must reach its absolute maximum allowable magnitude during the interval $t\in[t_0 \hspace{0.15cm} t_\text{CPA})$, where $t_\text{CPA}$ is the time of the closest point of approach within $[t_0 \hspace{0.15cm}t_f]$. The illustration of nuisance and nuisance-free activations is depicted in Fig.~\ref{fig:nuisanceVSnuisanceFree}.

\begin{figure}[hbt!]
\centering
\includegraphics[width=0.6\textwidth]{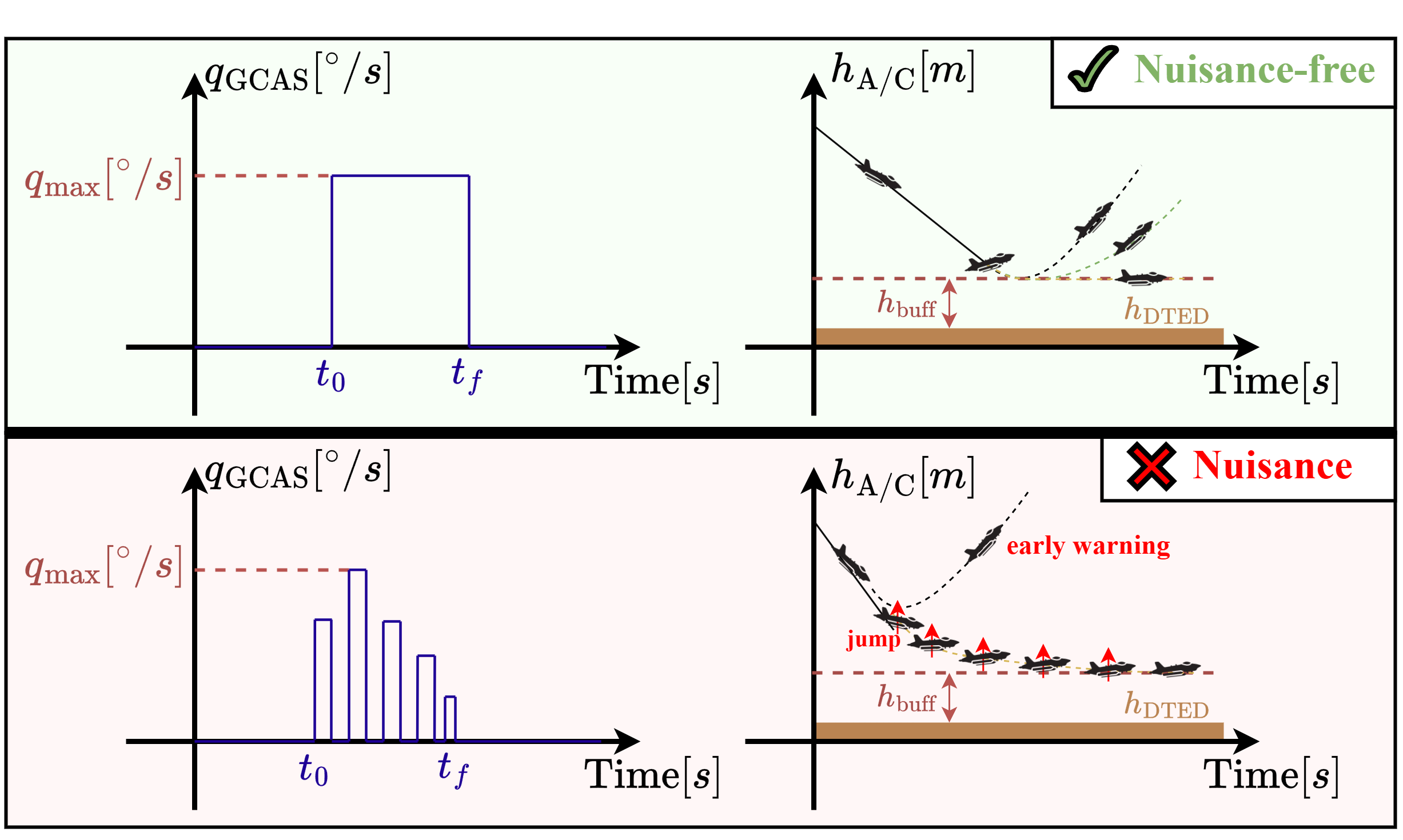}
\caption{The maneuvers and control inputs, classified as nuisance or nuisance-free, involve $q_\text{GCAS}$, the Auto-GCAS pitch rate command limited by $q_\text{max}$.}
\label{fig:nuisanceVSnuisanceFree}
\end{figure}

Although the primary focus is on achieving a nuisance-free design, the requirements of a complete Auto-GCAS framework extend further, necessitating adherence to the critical ``do no harm" principle, which encompasses a broad range of additional considerations \cite{runtimeMagazine}. Fundamentally, the Auto-GCAS must command a recovery maneuver which \textbf{(1)} does not damage pilot or the structural integrity of the aircraft, \textbf{(2)} does not stimulate an uncontrollable flight, and \textbf{(3)} allows the pilot to intervene \cite{runtimeMagazine, Agcas7, Agcas5}. In summary, the design of a comprehensive and complete Auto-GCAS framework must integrate both the nuisance-free and ``do no harm" principles, along with their associated sub-requirements. Thus, the term ``complete Auto-GCAS framework" in this study refers to a design that considers these criteria.

\subsection{Related Work}

Based on the discussed challenges, there is a limited amount of unclassified research addressing the Auto-GCAS design issue. Even though the mid-air or traffic collision avoidance problem has been extensively studied\textemdash for example, for quadrotors \cite{ca2}, fixed-wing manned \cite{ca4} and unmanned aircraft \cite{ca1, ca5}, and even for geofencing \cite{ca3}\textemdash their proposed methods cannot be directly applied to Auto-GCAS design due to the specific design criteria discussed earlier, which were not considered in those studies, with the exception of \cite{ca3}. Molnar et al. \cite{ca3} proposed a run-time assurance control framework for fixed-wing aircraft to ensure safety-critical tasks, such as mid-air collision avoidance and geofencing. They introduced various control barrier function (CBF) architectures, including higher-order, backstepping-based, and model-free methods, to address these issues. By merging multiple safety constraints into a unified CBF, they demonstrated the efficacy of their proposals through various simulations. However, as they highlighted as a direction for future work, the most significant limitation of their study may be the lack of additional safety constraints related to the aircraft's safe flight envelope\textemdash e.g. measures to prevent the aircraft from exceeding the stall angle of attack\textemdash indicating that their methodology does not incorporate the ``do no harm" requirements. Therefore, the available collision avoidance studies offer neither directly adoptable approaches nor a complete framework for the design of Auto-GCAS. Furthermore, in the existing accessible Auto-GCAS studies, the main focus area is generally nuisance-free design. In \cite{Agcas7}, the conceptual architecture of a nuisance-free Auto-GCAS was elaborated. The discussed architecture includes several key components, such as trajectory prediction algorithm to eliminate the potential nuisance situations. This component principally answers the question ``What would the resulting trajectory look like if a recovery maneuver is initiated now?" and terrain scan pattern that scans a virtual horizon using digital terrain elevation data (DTED). Seemingly, as one of the earliest publications on nuisance-free Auto-GCAS design, the study in \cite{Agcas7} laid a foundational framework for subsequent research, such as \cite{Agcas1, Agcas2, Agcas3, Agcas5, Agcas6, Agcas8, Agcas10, Agcas12}. As a follow-up to \cite{Agcas7}, the study in \cite{Agcas5} also emphasized the necessity of a trajectory prediction algorithm. Similarly, the study in \cite{Agcas3} presents a detailed framework for Auto-GCAS design, again including a trajectory prediction algorithm. Notably, several studies, including \cite{Agcas1, Agcas2, Agcas8, Agcas10}, focus specifically on trajectory prediction algorithms. For instance, \cite{Agcas1} proposed a ground collision avoidance warning and decision system with multi-trajectory risk assessment and decision functions to provide comprehensive avoidance decisions for flight crews. Likewise, \cite{Agcas8} and \cite{Agcas10} introduced trajectory prediction algorithms tailored for general aviation (GA). Finally, the study in \cite{Agcas11} proposed a nuisance-free Auto-GCAS design approach by formulating timely and aggressive recovery maneuvers as an optimal control problem, thereby eliminating the need for a trajectory prediction algorithm. However, this formulation was based on simplified point-mass, three degrees-of-freedom (3DoF) flight dynamics.

Consequently, most existing Auto-GCAS approaches in the literature rely on computationally intensive trajectory prediction algorithms to satisfy the nuisance-free criterion. This reliance typically limits implementations to simplified 3DoF kinematic models, restricting their applicability to realistic flight dynamics. Moreover, these methods often lack the flexibility to adapt to varying flight conditions or enforce dynamic constraints in real time. To the best of the authors’ knowledge, there is currently no unclassified study that offers a complete Auto-GCAS framework capable of simultaneously addressing nuisance avoidance, dynamic feasibility, and real-time adaptability within a unified and certifiable control structure.

\subsection{Objectives \& Methodology}
\label{objectivesMethodology}

Building on the discussed design challenges and existing literature, a pertinent research question arises: ``Can a computationally efficient alternative to trajectory prediction-based Auto-GCAS be developed that not only guarantees nuisance-free behavior but also provides a complete safety framework meeting all critical operational and dynamic constraints?" To answer this, the primary objective of this study is to develop an Auto-GCAS that guarantees nuisance-free operation while adhering to stringent safety requirements, including compliance with flight envelope protection principles to meet the “do no harm” mandate. To achieve this goal, this study proposes a novel design methodology that incorporates the following key components:

\begin{itemize}
\item \textbf{Exponential Control Barrier Functions (ECBF):} To establish the ground collision avoidance constraint by considering the altitude dynamics of the aircraft.
\item \textbf{Adaptive Sliding Manifolds:} To manipulate the ECBF dynamics and ensure nuisance-free operation, tailored to critical flight parameters, such as pitch angle, bank angle, and true airspeed, for adaptation to varying conditions.
\end{itemize}

Additionally, a flight envelope protection algorithm using control barrier functions is designed to validate the commands generated by the Auto-GCAS, ensuring compliance with flight envelope protection constraints, including angle of attack, load factor, and bank angle. The integration of CBF/ECBF is motivated by their advantages: \textbf{(1)} the provision of rigorous mathematical safety guarantees, \textbf{(2)} the facilitation of control authority-sharing between humans and automation through their formulations, and \textbf{(3)} the computational efficiency through linear constrained convex optimization, enabling real-time feasibility. Therefore, in contrast to trajectory planning or prediction-based approaches, CBFs require only the current state and allow online enforcement of safety without full and/or partial horizon planning. The use of exponential CBFs further enhances responsiveness near constraint boundaries, and their combination with adaptive manifolds enables tuning of intervention aggressiveness based on pitch angle, bank angle, and true airspeed.

Eventually, the proposed methodology is validated through various ground collision scenarios and Monte Carlo simulations, emphasizing \textbf{(1)} authority-sharing between human and automated controls, \textbf{(2)} collision avoidance capabilities, and \textbf{(3)} nuisance-free operation, minimizing false alarms or unnecessary interventions.

\subsection{Contributions \& Organization}
\label{contributionsOrganization}

The contributions of the study are itemized as follows;

\begin{itemize}
    \item Contrary to what has been proposed in \cite{Agcas1, Agcas2, Agcas3, Agcas5, Agcas6, Agcas8, Agcas10, Agcas12}, this research does not rely on computationally intensive trajectory prediction algorithms to generate nuisance-free recovery maneuvers. Instead, it derives an exponential control barrier function (ECBF)-based ground collision avoidance constraint. The ECBF dynamics are manipulated using adaptive sliding manifolds, which are designed through an offline optimization framework to meet the nuisance-free requirement.
    \item To achieve a complete safety framework, this study also designs a flight envelope protection (FEP) algorithm using control barrier functions (CBFs), which is similar to the previous work of the authors \cite{altunkaya2}. The FEP algorithm supervises Auto-GCAS commands to ensure compliance with critical flight envelope limits, such as the stall angle of attack and maximum load factor. This comprehensive framework addresses ``do no harm" requirement, which is a gap observed in \cite{ca3}.
    \item The proposed method achieves a success rate of $849$ out of $850$ Monte Carlo simulations for random initial dive cases, equating to $99.88\%$. This success rate is attributed to the study's formally safety-proven approach, highlighting the method's potential to significantly reduce controlled flight into terrain (CFIT) incidents.
\end{itemize}

The rest of the paper is organized as follows: In Section~\ref{preliminaries}, the necessary background for the study is provided, including nonlinear flight dynamics modeling and flight control law design. Section~\ref{groundCollisionAvoidanceSystem} introduces the design of the ground collision avoidance system, addressing nuisance-free design in Section~\ref{nuisanceFreeDesign}. To complete the safety framework, the flight envelope protection design is presented in Section~\ref{flightEnvelopeProtection}, with specific discussions on angle of attack and load factor in Section~\ref{AoALoadFactor}, and bank angle in Section~\ref{bankAngle}. Subsequently, the proposed architecture is rigorously assessed in Section~\ref{results}, through ground collision scenarios in Section~\ref{terrainCollisionResults}, pilot authority-sharing scenarios in Section~\ref{authoritySharingResults}, and Monte Carlo simulations, including $850$ different cases, in Section~\ref{monteCarloResults}. Finally, a brief discussion is provided in Section~\ref{discussion} to elaborate on the potential of the proposed method.

\section{Preliminaries}
\label{preliminaries}

The necessary background for the subsequent sections is established in this section.

\subsection{Notations}
\label{notations}

A vector is denoted in the bold type, i.e. $\bm{v}$, and the vector product of two vectors $\bm{x}$ and $\bm{y}$ is denoted by $\bm{x} \times \bm{y}$, and $\text{sgn}(*)$ is the signum function. Throughout the study, the prefix $\Delta$ denotes the incremental form, i.e. $\Delta (*)$ is the incremental form of $(*)$. The notation of $s_{*}$, $c_{*}$, and $t_{*}$ corresponds to sine, cosine, and tangent of $(*)$. Finally, a control affine system is described as given in Eq.~\eqref{controlAffineSystem}.

\begin{equation}
\label{controlAffineSystem}
    \bm{\dot{x}} = \bm{f}(\bm{x}) + \bm{g}(\bm{x})\bm{u},
\end{equation}  
where $\bm{x} \in \mathbb{R}^n$ is the state vector, and $\bm{u} \in \mathbb{R}^m$ is the control input vector. Nonlinear mappings of $\bm{f} : \mathbb{R}^n \to \mathbb{R}^n$ and $\bm{g} : \mathbb{R}^n \to \mathbb{R}^{n \times m}$ are locally Lipschitz continuous functions.

\subsection{Control Barrier Functions}
\label{exponentialControlBarrierFunction}

To formally define safety constraints in control systems, let define the notion of a barrier function $b(\bm{x})$. The function $b(\bm{x}): D \subset \mathbb{R}^n \to \mathbb{R}$ is designed such that a safe set $C$ can be defined as:
\begin{equation}
    C = \{\bm{x} \in D \subset \mathbb{R}^n \mid b(\bm{x}) \geq 0 \},
\end{equation}
where $b(\bm{x}) \geq 0$ ensures that the system state remains within the safe set, $\partial C = \{ x \in D \mid b(\bm{x}) = 0 \}$ represents the boundary, and $\text{Int}(C) = \{ x \in D \mid b(\bm{x}) > 0 \}$ denotes the interior of the safe set.

Given a continuously differentiable function $b(\bm{x})$ and a dynamical system in Eq.~\eqref{controlAffineSystem}, the Lie derivatives of $b(\bm{x})$ along the vector fields $\bm{f}(\bm{x})$ and $\bm{g}(\bm{x})$ are defined as:
\begin{equation}
    \mathcal{L}_f b(\bm{x}) = \frac{\partial b(\bm{x})}{\partial \bm{x}} \bm{f}(\bm{x}), \quad \mathcal{L}_g b(\bm{x}) = \frac{\partial b(\bm{x})}{\partial \bm{x}} \bm{g}(\bm{x}).
\end{equation}

The relative degree of $b(\bm{x})$ with respect to the control input $\bm{u}$ is the number of times it must be differentiated along the system dynamics before $\bm{u}$ explicitly appears. For a function $b(\bm{x})$ with relative degree $\delta$, its higher-order Lie derivatives are given by:
\begin{equation}
    \mathcal{L}_f^\delta b(\bm{x}) = \frac{d^\delta}{dt^\delta} b(\bm{x}),
\end{equation}
where $\mathcal{L}_f^\delta b(\bm{x})$ denotes the Lie derivative of $b(\bm{x})$ to the power of its relative degree.

\begin{definition}[Exponential Control Barrier Function]
A continuously differentiable function $b(\bm{x}): D \subset \mathbb{R}^n \to \mathbb{R}$ is an exponential control barrier function (ECBF) if there exists a coefficient gain vector $\kappa \in \mathbb{R}^\delta$ such that for all $\bm{x} \in C$:
\begin{equation}
\label{eCBF-constraint}
    \sup_{\bm{u} \in U} \big\{ \mathcal{L}^\delta_f b(\bm{x}) + \mathcal{L}_g \mathcal{L}^{\delta-1}_f b(\bm{x}) \bm{u} + \kappa \eta_b(\bm{x}) \big\} \geq 0,
\end{equation}
where $\eta_b(\bm{x}) = [b(\bm{x}), \mathcal{L}_f b(\bm{x}), ..., \mathcal{L}_f^{\delta-1} b(\bm{x})]^\mathrm{T}$ is the Lie derivative vector of $b(\bm{x})$ and $\kappa = [k_0, k_1, \dots, k_{\delta-1}]$ denotes the coefficient gain vector associated with $\eta_b(\bm{x})$.
\end{definition}

\subsection{Flight Dynamics Model}
\label{fdm}

The baseline aircraft considered is an over-actuated F-16, featuring five independent control surface actuators: the right and left horizontal tails, right and left ailerons, and the rudder. Consequently, the aerodynamic modeling and flight control law design are specifically adapted to account for this over-actuated configuration, as detailed in the subsequent sections.

\subsubsection{Equations of motion}
\label{equationsOfMotion}

The axes frame of the baseline aircraft is depicted in Fig.~\ref{f16axesFrame}.

\begin{figure}[hbt!]
\centering
\includegraphics[width=0.6\textwidth]{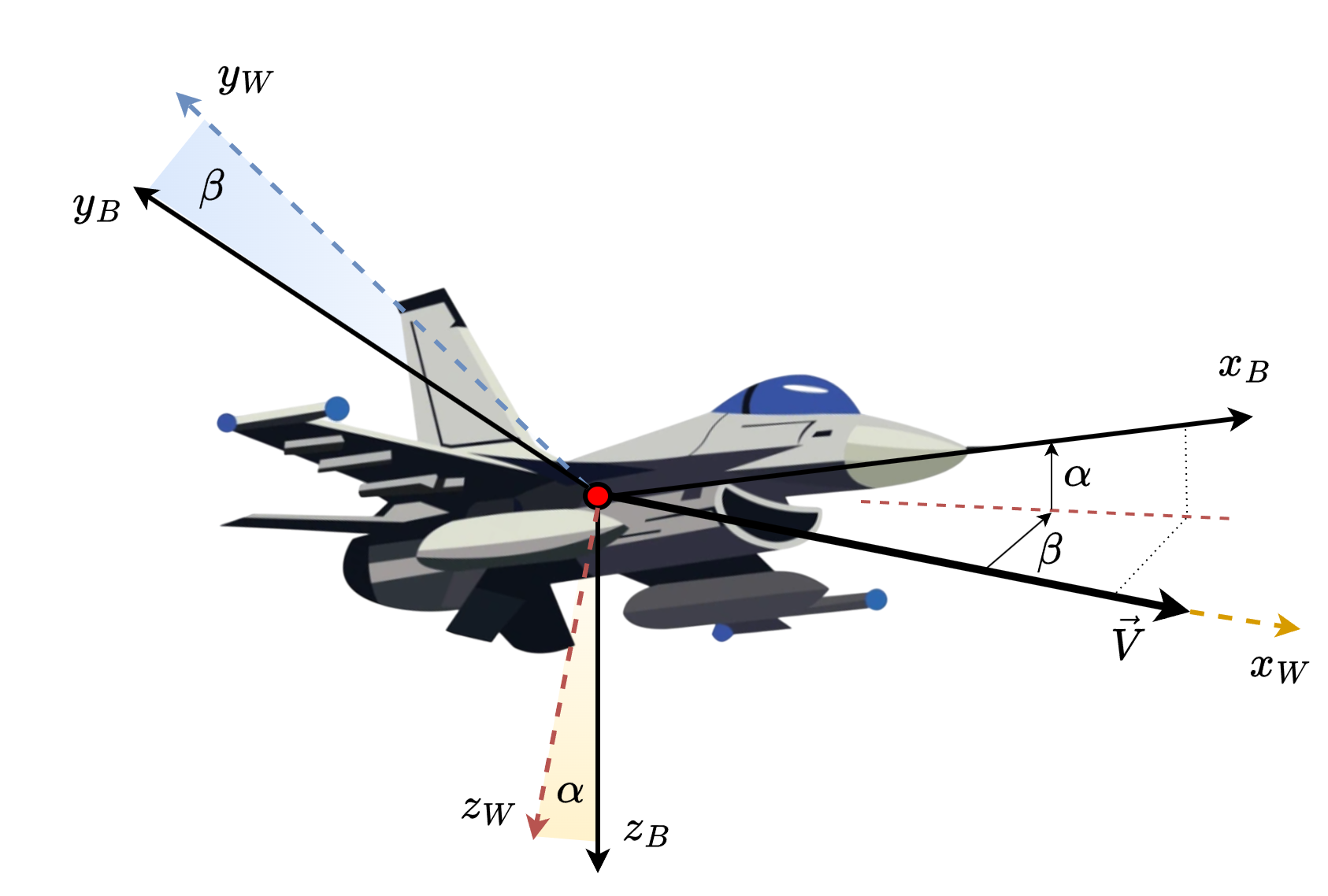}
\caption{An illustration of the baseline aircraft with its body axis (denoted as $b$) and wind axis (denoted as $w$) frames.}
\label{f16axesFrame}
\end{figure}

The set of nonlinear flight dynamics equations is presented in Eq.~\eqref{flightDynamicsEq}, including translational and rotational dynamics, and translational and rotational kinematics, respectively.

\begin{equation}
\label{flightDynamicsEq}
\begin{split}
&\begin{cases}
  \dot{u} = \sum F_x/m + rv - qw \\
  \dot{v} = \sum F_y/m + pw - ru \\
  \dot{w} = \sum F_z/m + qu - pv
\end{cases}
\\
&\begin{cases}
  \dot{p} = qr(I_\text{yy} - I_\text{zz})/I_\text{xx} + (\dot{r} + pq)I_\text{xz}/I_\text{xx} + \sum L/I_\text{xx} \\
  \dot{q} = pr(I_\text{zz} - I_\text{xx})/I_\text{yy} + (r^2 - p^2)I_\text{xz}/I_\text{yy} + \sum M/I_\text{yy} \\
  \dot{r} = pq(I_\text{xx} - I_\text{yy})/I_\text{zz} + (\dot{p} - qr)I_\text{xz}/I_\text{zz} + \sum N/I_\text{zz}
\end{cases}
\\
&\begin{cases}
  \dot{x}_\text{E} = u c_\theta c_\psi + v(s_\phi s_\theta c_\psi - c_\phi s_\psi) + w(c_\phi s_\theta c_\psi + s_\phi s_\psi) \\
  \dot{y}_\text{E} = u c_\theta c_\psi + v(s_\phi s_\theta c_\psi + c_\phi c_\psi) + w(c_\phi s_\theta s_\psi - s_\phi s_\psi) \\
  \dot{z}_\text{E} = -u s_\theta + v s_\phi c_\theta + w c_\phi c_\theta
\end{cases}
\\
&\begin{cases}
  \dot{\phi} = p + t_\theta(q s_\phi + r c_\phi) \\
  \dot{\theta} = q c_\phi - r s_\phi \\
  \dot{\psi} = (q s_\phi + r c_\phi) / c_\theta
\end{cases}
\\
\end{split}
\end{equation}
where $[u \hspace{0.15cm} v \hspace{0.15cm} w]^\mathrm{T} \in \mathbb{R}^3$ represent the body velocity components, while $[p \hspace{0.15cm} q \hspace{0.15cm} r]^\mathrm{T} \in \mathbb{R}^3$ are the body angular rate components. The navigational position components are denoted as $[x_\text{E} \hspace{0.15cm} y_\text{E} \hspace{0.15cm} z_\text{E}]^\mathrm{T} \in \mathbb{R}^3$, and $[\phi \hspace{0.15cm} \theta \hspace{0.15cm} \psi]^\mathrm{T} \in \mathbb{R}^3$ represent the Euler angles. The force components acting on the body frame are $[F_x \hspace{0.15cm} F_y \hspace{0.15cm} F_z]^\mathrm{T} \in \mathbb{R}^3$, and $[L \hspace{0.15cm} M \hspace{0.15cm} N]^\mathrm{T} \in \mathbb{R}^3$ correspond to the roll, pitch, and yaw moments. Finally, $m$ denotes the mass of the aircraft, and $I_\text{xx}$, $I_\text{yy}$, $I_\text{zz}$, and $I_\text{xz}$ are the moment of inertia components of inertia tensor $J \in \mathbb{R}^{3 \times 3}$ of the aircraft.

\subsubsection{Aerodynamics \& Actuators}
\label{aerodynamicsActuators}

The over-actuated aerodynamic model is based on the methodology presented in \cite{aerodynamicPolynomials, forceMomentEquations}, where aerodynamic coefficients are expressed as polynomial functions of the relevant states and control surface deflections. Comprehensive details of this formulation are available in \cite{aerodynamicPolynomials}. Additionally, the actuator dynamics are represented by a first-order system incorporating time constants, rate limits, and position saturation constraints, as described in \cite{F16}. 

\subsection{Flight Control Law Design}  
\label{flightControlLaw}  

The flight control law is composed of two primary elements: \textbf{(1)} a control augmentation system (CAS) implementing a single-loop angular rate control based on nonlinear dynamic inversion (NDI), and \textbf{(2)} an incremental nonlinear control allocation (INCA) to handle over-actuation.

\subsubsection{Control Augmentation System Design}  
\label{controlAugmentationSystem}  

The derivation of the control law leverages the control-affine structure of Euler's equations of motion, represented in a decomposed form in Eq.~\eqref{rotationalDynamics2}.

\begin{equation}
\label{rotationalDynamics2}
    \underbrace{\bm{\dot{\omega}}}_{\substack{\dot{\bm{x}}}} =  
    \underbrace{-J^{-1} (\bm{\omega} \times J \bm{\omega})}_{\substack{\bm{f}(\bm{x})}} +  
    \underbrace{J^{-1}\bar{q}_\infty S  
    \begin{bmatrix}
      b &  &  \\ 
      & \bar{c} &  \\ 
      &  & b \\ 
    \end{bmatrix}}_{\substack{\bm{g}(\bm{x})}}  
    \underbrace{\bm{\tau}}_{\substack{\bm{u}}},  
\end{equation}  
where $\bm{\tau} \in \mathbb{R}^3$, $\Bar{q}_\infty \in \mathbb{R}$, $S \in \mathbb{R}$, $b \in \mathbb{R}$, and $\Bar{c} \in \mathbb{R}$ represent the aerodynamic moment coefficient vector, dynamic pressure, wing area, wing span, and mean aerodynamic chord, respectively. For a control-affine system as defined in Eq.~\eqref{controlAffineSystem}, the NDI control law is expressed by Eq.~\eqref{NDIlaw1}.

\begin{equation}
\label{NDIlaw1}
    \bm{u} = \bm{g}(\bm{x})^{-1}[\bm{\nu} - \bm{f}(\bm{x})],
\end{equation}  
where $\bm{\nu} = \dot{\bm{x}}_c$ is the virtual input designed using a linear controller. The control moment coefficients for angular rate regulation are then derived from Eq.~\eqref{NDIlaw1}, as presented in Eq.~\eqref{NDIlaw2}.

\begin{equation}
\label{NDIlaw2}
    \bm{\tau}_c = \Bigg\{J^{-1}\bar{q}_\infty S  
    \begin{bmatrix}
      b &  &  \\ 
      & \bar{c} &  \\ 
      &  & b \\ 
    \end{bmatrix} \Bigg\}^{-1}\Big[\bm{\dot{\omega}}_c + J^{-1} (\bm{\omega} \times J \bm{\omega})\Big].
\end{equation}  

The virtual input $\bm{\dot{\omega}}_c \in \mathbb{R}^3$ is given by Eq.~\eqref{eq:fca13}.

\begin{equation} \label{eq:fca13}
\bm{\dot{\omega}}_c =  
\underbrace{\begin{bmatrix}
K_p & & \\ 
 & K_q & \\ 
 & & K_r \\ 
\end{bmatrix}}_{\substack{\overset{\Delta}{=} K}}  
\underbrace{\begin{bmatrix}
p_\text{pilot} - p \\ 
q_\text{pilot} - q \\ 
r_\text{pilot} - r \\ 
\end{bmatrix}}_{\substack{\overset{\Delta}{=} e}},
\end{equation}  
where $K \overset{\Delta}{=} \{K_p, K_q, K_r\}$ denotes the gain matrix for roll, pitch, and yaw channels, respectively. These formulations provide the required control moment coefficients in response to pilot commands $p_\text{pilot}\in \mathbb{R}$, $q_\text{pilot}\in \mathbb{R}$, and $r_\text{pilot}\in \mathbb{R}$. The control moment coefficients are subsequently transferred to the control allocation module, detailed in the next section.

\subsubsection{Control Allocation Design}  
\label{controlAllocation}  

Incremental nonlinear control allocation (refer to \cite{allocation1, allocation2} for further details) is defined by Eq.~\eqref{eq:fca14}.

\begin{equation} \label{eq:fca14}
    \Delta \bm{\delta}_c = \Phi^{-1}(\bm{x}_0, \bm{\delta}_0) \Delta \bm{\tau}_c,
\end{equation}  
where $\Delta \bm{\delta}_c = \bm{\delta}_c - \bm{\delta}_0$ and $\Delta \bm{\tau}_c = \bm{\tau}_c - \bm{\tau}_0$, with the subscript “$0$” denoting the current state. The control effectivity matrix $\Phi \in \mathbb{R}^{3 \times n}$, shown in Eq.~\eqref{eq:fca15}, contains moment coefficient derivatives with respect to control surface deflections at the current state, where $n$ represents the number of control surfaces. The control surface deflections $\bm{\delta} \in \mathbb{R}^5$ correspond to the right and left horizontal tails, right and left ailerons, and the rudder.

\begin{equation} \label{eq:fca15}
\Phi(\bm{x}_0, \bm{\delta}_0) =  
\begin{bmatrix}
{\dfrac{\partial C_l}{\partial \delta_1}}\Big|_{(\bm{x}_0, \bm{\delta}_0)} & \cdots & {\dfrac{\partial C_l}{\partial \delta_n}}\Big|_{(\bm{x}_0, \bm{\delta}_0)} \\ 
{\dfrac{\partial C_m}{\partial \delta_1}}\Big|_{(\bm{x}_0, \bm{\delta}_0)} & \cdots & {\dfrac{\partial C_m}{\partial \delta_n}}\Big|_{(\bm{x}_0, \bm{\delta}_0)} \\ 
{\dfrac{\partial C_n}{\partial \delta_1}}\Big|_{(\bm{x}_0, \bm{\delta}_0)} & \cdots & {\dfrac{\partial C_n}{\partial \delta_n}}\Big|_{(\bm{x}_0, \bm{\delta}_0)} \\ 
\end{bmatrix}.
\end{equation}  

With five independent control surfaces ($n = 5$), $\Phi$ is a non-square matrix and can only be inverted using the Moore-Penrose pseudo-inverse ($\Phi^\ddag$). The final control surface deflections responding to the control moment coefficients $\bm{\tau}_c$ are determined by Eq.~\eqref{eq:fca17}.

\begin{equation} \label{eq:fca17}
    \bm{\delta}_c = \Phi^\ddag \Delta \bm{\tau}_c + \bm{\delta}_0.
\end{equation}  

These equations complete the flight control law for angular rate control.

\section{Ground Collision Avoidance System Design}
\label{groundCollisionAvoidanceSystem}

The principal strategy of ground collision avoidance should be to redesign the pilot commands of $p_\text{pilot}$ and $q_\text{pilot}$ since a rolling and pitching maneuver is expected. Thus, the altitude dynamics of the aircraft, which is $\dot{h} = -\dot{z}_E$, given by Eq.~\eqref{flightDynamicsEq}, should be such decomposed that the roll rate $p$ and pitch rate $q$ should be observable. However, as the pitch rate is the principal governing factor for altitude dynamics, the roll rate is reserved for use in the other recovery maneuver, i.e. bank-to-level, which will be scrutinized in the proceeding section. Therefore, this section focuses on defining the dynamics to isolate the pitch rate. Then, design a barrier function as presented in Eq.~\eqref{barrierH1}.

\begin{equation}
\label{barrierH1}
    b(h) = h - (h_\text{buff} + h_\text{DTED}).
\end{equation}

It is obvious that $b(h) > 0, \forall h \in (h_\text{buff} + h_\text{DTED}, \infty)$, and $b(h) = 0 \iff h = h_\text{buff} + h_\text{DTED}$. Thus, the time derivative of the barrier function is in Eq.~\eqref{barrierH2}.

\begin{equation}
\label{barrierH2}
    \dot{b}(h) = \dot{h} = \mathcal{L}_f b(h).
\end{equation}

However, in the first time derivative, desired input does not appear, thereby the second time derivative of the barrier function is presented in Eq.~\eqref{barrierH3}.

\begin{equation}
\begin{split}
\label{barrierH3}
    \ddot{b}(h) = \ddot{h} &= f(h) + g(h)u \\
    &= \mathcal{L}^2_f b(h) + \mathcal{L}_g \mathcal{L}_f b(h) u.
\end{split}
\end{equation}

The roll and pitch rates can be explicitly obtained by expanding Eq.~\eqref{barrierH3}; however, for simplicity, the following assumption in \ref{assH1} is both convenient and consistent.

\begin{assumption}
\label{assH1}
    We assume that the bank angle ($\phi$) and sideslip angle ($\beta$) are both $0^\circ$, as wing-level and symmetric flight is expected during the ground collision avoidance maneuver.
\end{assumption}

Thereby, the altitude dynamics of the aircraft turns into the following form given in Eq.~\eqref{barrierH4}.

\begin{equation}
\label{barrierH4}
    \dot{h} = V_T s_\gamma,
\end{equation}
where $V_T = ||\bm{V}||_2$, note that $\bm{V} = [u \hspace{0.15cm} v \hspace{0.15cm} w]^\mathrm{T} \in \mathbb{R}^3$, is the true airspeed, and $\gamma = \theta - \alpha$ due to the wings-level symmetric flight. Consequently, the altitude dynamics are reduced to dependencies on true airspeed, pitch angle, and angle of attack. At this stage, these dynamics should be decomposed to isolate the pitch rate. Based on this rationale, the second time derivative of the barrier function to observe the pitch rate is given by Eq.~\eqref{barrierH5}.

\begin{equation}
\label{barrierH5}
\begin{aligned}
    \ddot{h} &= \dot{V}_T s_{(\theta - \alpha)} + V_T c_{(\theta - \alpha)}\big(\dot{\theta} - \dot{\alpha} \big) \\
    &= \big[f(V_T) + g(V_T)q \big] s_{(\theta - \alpha)} + V_T c_{(\theta - \alpha)} \big[f(\theta) + g(\theta)q - f(\alpha) - g(\alpha)q \big]. \\
\end{aligned}
\end{equation}

Since true airspeed, pitch rate, and angle of attack dynamics involve the pitch rate in their first derivative, it is appropriate to describe them in the form of $\dot{V}_T = f(V_T) + g(V_T)q$, $\dot{\theta} = f(\theta) + g(\theta)q$, and $\dot{\alpha} = f(\alpha) + g(\alpha)q$, respectively. Afterwards, for the sake of clarity, the description of the altitude dynamics in the form of $\ddot{h} = f(h) + g(h)q$ is presented in Eq.~\eqref{barrierH6}.

\begin{equation}
\label{barrierH6}
\begin{aligned}
    \ddot{h} &= \underbrace{f(V_T) + V_T c_{(\theta - \alpha)} \big[f(\theta) - f(\alpha)\big]}_{\substack{f(h)}} + \underbrace{\Big(g(V_T)s_{(\theta - \alpha)} + V_T c_{(\theta - \alpha)} \big[g(\theta) - g(\alpha)\big] \Big)}_{\substack{g(h)}} q.
\end{aligned}
\end{equation}

At this point, the dynamics of true airspeed, pitch angle, and angle of attack should be defined; however, the detailed derivation of these dynamics is not included in this study, as it is readily available in the flight dynamics and control literature, e.g. \cite{snell, altunkaya1, altunkaya3}. The true airspeed dynamics are given by Eq.~\eqref{barrierH7}.

\begin{equation}
\label{barrierH7}
\begin{aligned}
    \dot{V}_T &= \dfrac{1}{m}\big[-Dc_\beta + Cs_\beta + Tc_\alpha c_\beta - mg\big( s_\theta c_\alpha c_\beta - c_\theta s_\phi s_\beta - c_\theta c_\phi s_\alpha c_\beta \big) \big],
\end{aligned}
\end{equation}
where $m$ is the mass of the aircraft, $g$ is gravity, $D$, $C$, and $T$ denote the drag, cross, and thrust forces, respectively. A further decomposition for the drag force can be performed, as presented in Eq.~\eqref{dragDecomposition}.

\begin{equation}
\label{dragDecomposition}
    \dfrac{-D c_\beta}{m} = \dfrac{-\bar{q}_\infty S C_{D_q} \dfrac{q\bar{c}}{2V_T} c_\beta}{m} + \dfrac{-\bar{q}_\infty S C_{D_\text{aux}} c_\beta}{m}.
\end{equation}

This decomposition should be done to reveal hidden pitch rate contributions within the aerodynamic coefficients, where $C_{D_\text{aux}}$ represents remaining lift coefficient terms excluding $C_{D_q}$. For the simplicity, the assumption in \ref{assH1} also leads to the following assumption described in \ref{assH2}.

\begin{assumption}
\label{assH2}
    For the small magnitudes of sideslip angle ($\beta$), the contribution of the term of $C s_\beta$ becomes negligible. Note that the sideslip angle is expected to be $0^\circ$ during a symmetric level flight; therefore, this assumption is convenient and consistent.
\end{assumption}

Finally, one can easily obtain the components of true airspeed dynamics, i.e. $f(V_T)$ and $g(V_T)$, as given by Eq.~\eqref{barrierH8}.

\begin{equation}
\label{barrierH8}
\begin{aligned}
    f(V_T) &= \frac{T c_\alpha c_\beta - mg\big(s_\theta  c_\alpha  c_\beta - c_\theta s_\phi s_\beta - c_\theta c_\phi s_\alpha c_\beta \big)}{m}  - \dfrac{\bar{q}_{\infty} S C_{D_\text{aux}} c_\beta}{m}, \\
    g(V_T) &= -\dfrac{\bar{q}_\infty S C_{D_q} \dfrac{\bar{c}}{2V_T} c_\beta}{m}.
\end{aligned}
\end{equation}

As the next step, define the pitch angle dynamics as given in Eq.~\eqref{barrierH9}.

\begin{equation}
\label{barrierH9}
    \dot{\theta} = q c_\phi - r s_\phi.
\end{equation}

It is quite easy to decompose the components of pitch angle dynamics, as given by Eq.~\eqref{barrierH10}.

\begin{equation}
\label{barrierH10}
\begin{aligned}
    &f(\theta) = -rs_\phi, \\
    &g(\theta) = c_\phi.
\end{aligned}
\end{equation}

The final step is the introducing the angle of attack dynamics, given in Eq.~\eqref{barrierH11}.

\begin{equation}
\label{barrierH11}
\begin{aligned}
    \dot{\alpha} &= \dfrac{-L}{m V_T c_\beta} + \dfrac{mg\big(c_\theta c_\phi c_\alpha + s_\theta s_\alpha \big) - T}{m V_T c_\beta} + q - t_\beta \big(p c_\alpha + r s_\alpha \big),
\end{aligned}
\end{equation}
where $L$ denotes the lift force. Again, a decomposition for the lift force can be applied to reveal the hidden pitch rate dynamics, as given by Eq.~\eqref{barrierH12}.

\begin{equation}
\label{barrierH12}
    \dfrac{-L}{m V_T c_\beta} = \dfrac{-\bar{q}_\infty S C_{L_q} \dfrac{q\bar{c}}{2V_T}}{mV_T c_\beta} + \dfrac{-\bar{q}_\infty S C_{L_\text{aux}} }{mV_T c_\beta},
\end{equation}
where $C_{L_\text{aux}}$ are remaining lift coefficient terms excluding $C_{L_q}$. Finally, the components of the angle of attack dynamics are presented in Eq.~\eqref{barrierH13}.

\begin{equation}
\label{barrierH13}
\begin{aligned}
    f(\alpha) &= \frac{m g \big(c_\theta c_\phi c_\alpha + s_\theta s_\alpha \big) - T}{m V_T c_\beta} - t_\beta \big(p c_\alpha + r s_\alpha \big) - \frac{\bar{q}_{\infty} S C_{L_\text{aux}}}{m V_T c_\beta}, \\
    g(\alpha) &= 1 - \dfrac{\bar{q}_\infty S C_{L_q} \frac{\bar{c}}{2V_T}}{m V_T c_\beta}.
\end{aligned}
\end{equation}

Then, the ECBF constraint for the Auto-GCAS is prepared for construction. In this regard, the $\eta_b(h)$ vector can be represented by Eq.~\eqref{barrierH14}.

\begin{equation}
\label{barrierH14}
\eta_b(h)
=
\begin{bmatrix}
b(h) \\
\mathcal{L}_f b(h) \\
\end{bmatrix} 
=
\begin{bmatrix}
h - (h_\text{buff} + h_\text{DTED}) \\
\dot{h} \\
\end{bmatrix}.
\end{equation}

Furthermore, the ECBF constraint for the altitude dynamics as pitch rate is the command of Auto-GCAS can be described in Eq.~\eqref{barrierH15}.

\begin{equation}
\label{barrierH15}
    \underbrace{f(h) + g(h)q_\text{GCAS}}_{\substack{\mathcal{L}^2_f b(h) + \mathcal{L}_g \mathcal{L}_f b(h) q}} + \kappa \eta_b(h) \geq 0,
\end{equation}
where $\kappa = [k_1 \hspace{0.15cm} k_2]$. Consequently, the final form of the Auto-GCAS formulation for generating a pitch rate command to protect the aircraft from the DTED with a buffer height is presented in Eq.~\eqref{barrierH16}.

\begin{equation}
\label{barrierH16}
    \begin{aligned}
    u^\star = \operatorname*{argmin}_{q_\text{GCAS} \in \mathbb{R}} \quad & \dfrac{1}{2}\big(q_\text{GCAS} - q_\text{pilot}\big)^2 \\
    \textrm{s.t.} \quad & f(h) + g(h) q_\text{GCAS} + \kappa \eta_b(h) \geq 0 \\
                \quad & q_\text{min} \leq q_\text{GCAS} \leq q_\text{max}. \\
    \end{aligned}
\end{equation}

The constructed formulation enables the Auto-GCAS to generate a pitch rate command ($q_\text{GCAS}$) that closely follows the pitch rate command of pilot ($q_\text{pilot}$) while adhering to the constraint for ground collision avoidance. Additionally, the generated pitch rate command, $q_\text{GCAS}$, must remain within the interval $[q_\text{min} \hspace{0.15cm} q_\text{max}]$, considering the admissible and allowable pitch rate limits depending on the aircraft. However, at this stage, the nuisance-free criterion has not been addressed, and the question of whether it is possible to generate nuisance-free commands by adjusting the gains of $\kappa$, i.e., $k_1$ and $k_2$, will be elaborated in the proceeding section.

\subsection{Nuisance-free Intervention Design}
\label{nuisanceFreeDesign}

It is clearly observable that the ECBF constraint, given in Eq.~\eqref{barrierH15}, leads to a second-order linear system as presented in Eq.~\eqref{nuisance1}, setting $h_\text{buff} + h_\text{DTED} = 0$ for the sake of homogeneity. 

\begin{equation}
\label{nuisance1}
    \underbrace{\mathcal{L}^2_f b(h) + \mathcal{L}_g \mathcal{L}_f b(h) q_\text{GCAS}}_{\substack{\ddot{h}}} +  \underbrace{\kappa \eta_b(h)}_{\substack{k_2 \dot{h} + k_1 h}} \geq 0.
\end{equation}

The system also represents the second-order altitude dynamics, and the state-space form of the obtained system is given in Eq.~\eqref{nuisance2}, with a new state vector defined as $\bm{\xi} = [\xi_1 \hspace{0.15cm} \xi_2]^\mathrm{T}$.

\begin{equation}
\label{nuisance2}
\begin{bmatrix}
\dot{\xi_1} \\
\dot{\xi_2} \\
\end{bmatrix} 
=
\underbrace{\begin{bmatrix}
0 & 1 \\
-k_1 & -k_2
\end{bmatrix}}_{\substack{\Phi(\kappa)}} 
\begin{bmatrix}
\xi_1 \\
\xi_2 \\
\end{bmatrix},
\end{equation}
where $\xi_1 = h$ and $\xi_2 = \dot{h}$. The obtained second-order linear system has a clearly observable unique equilibrium point at the origin. Thereby, the manipulation of the ECBF dynamics is quite easy by adjusting the gains $k_1$ and $k_2$ in such a way that sliding manifolds are established over the phase portrait of the ECBF dynamics. Nevertheless, the main question is the classification of the equilibrium point, namely whether it is a node, saddle-node, circle, spiral, or degenerate node. In this regard, the most appropriate option should definitely be an attractor point, considering the assurance of stability regardless of the initial points within the phase space. Therefore, the saddle-node, circle, and degenerate node options are eliminated. Moreover, a spiral equilibrium point would likely result in an Auto-GCAS command that could be regarded as nuisance, considering Fig.~\ref{fig:nuisanceVSnuisanceFree}, since an oscillatory behavior would be observable in altitude dynamics. Consequently, the equilibrium point is preferred to be an attractive node with a high damping ratio, i.e., $\zeta \geq 1$. As a final point of concern, it should be noted that the higher the damping ratio, the more sluggish the altitude dynamics become. In other words, a setting of $\zeta > 1$ may result in an untimely intervention. Therefore, as the final specification, the setting $\zeta = 1$ is preferred as the most convenient option. The illustration of the phase portrait of the system defined in Eq.~\eqref{nuisance2} is depicted in Fig.~\ref{phasePortrait}.

\begin{figure}[hbt!]
\centering
\includegraphics[width=0.6\textwidth]{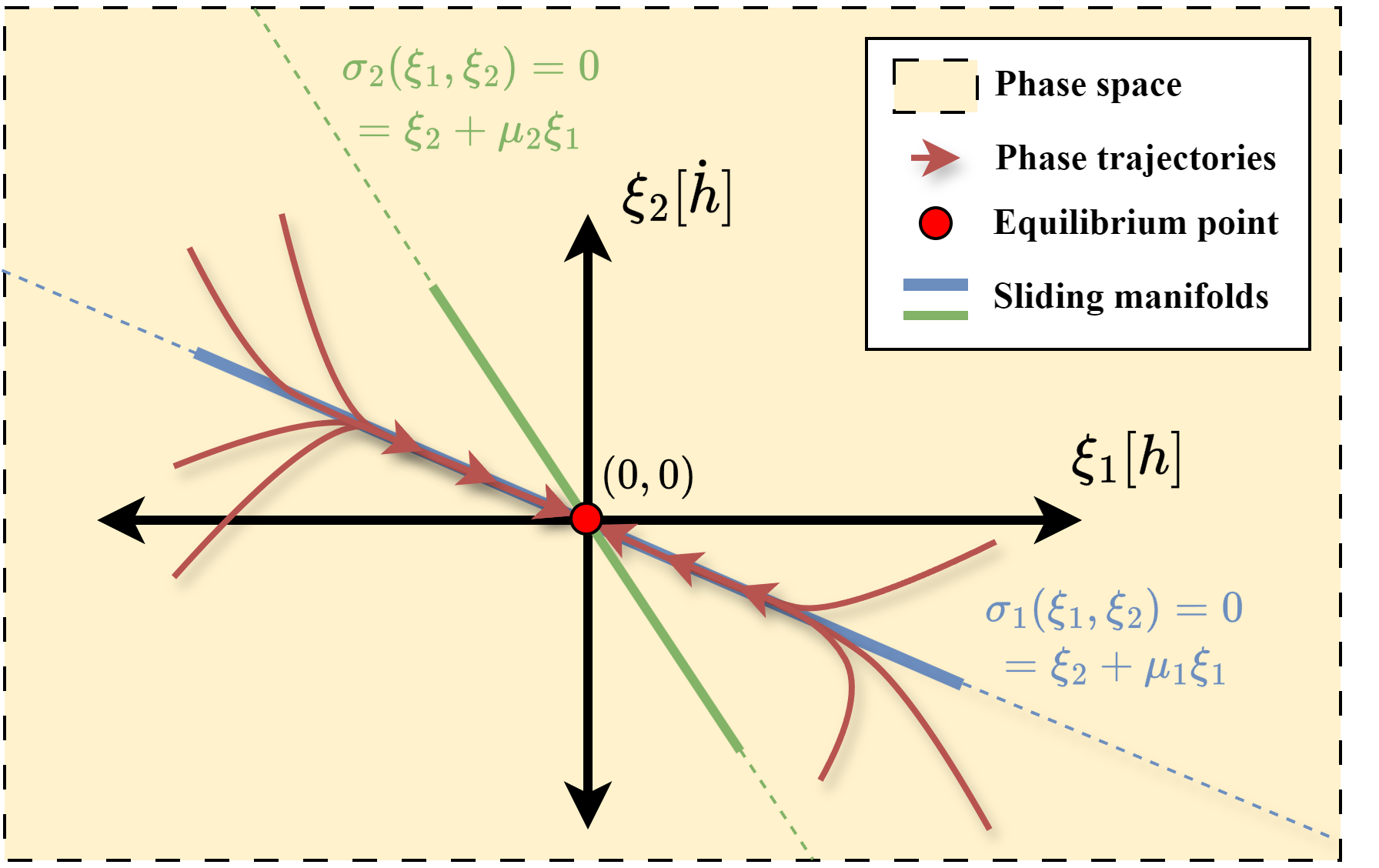}
\caption{An illustration of the phase portrait of the system in Eq.~\eqref{nuisance2}: representative linear sliding manifolds, $\sigma_1$ and $\sigma_2$, where $\mu = k_1/k_2$.}
\label{phasePortrait}
\end{figure}

Then, the proposition in \ref{prop1} must hold for the equilibrium point to be an attractive node.

\begin{proposition}
\label{prop1}
    The equilibrium point at the origin is an attractive node provided that $2\sqrt{k_1} \leq k_2$, where $k_1, k_2 \in \mathbb{R}_{>0}$.
\end{proposition}

\begin{proof}
    One can easily prove, noticing that $\text{Re}(\lambda(\Phi(\kappa))) < 0$ and $\text{Im}(\lambda(\Phi(\kappa))) = 0$, where $\lambda$ to be the eigenvalues of $\kappa$.
\end{proof}

Lastly, the proposition in \ref{prop2} must hold for the system dynamics to have a damping ratio $\zeta = 1$.

\begin{proposition}
\label{prop2}
    The damping ratio of the system dynamics, $\zeta = 1$, provided that $2\sqrt{k_2} = k_1$.
\end{proposition}

\begin{proof}
    One can easily prove, noticing that the damping ratio is $\zeta = \dfrac{k_1}{2\sqrt{k_2}}$ for a second order linear system given by Eq.~\eqref{nuisance2}.
\end{proof}

The determination of the characteristics of the equilibrium point and system dynamics is necessary but not sufficient. Rapidly varying flight states across a wide spectrum of operating conditions lead to the need for adaptivity of the gains $k_1$ and $k_2$ to ensure nuisance-free commands under any circumstances. Therefore, the sliding manifolds must be designed adaptively. As observed in Fig.~\ref{phasePortrait}, the linear sliding manifolds are functions of the gains $k_1$ and $k_2$, which then define an adaptive sliding manifold, as represented in Eq.~\eqref{nuisance3}.

\begin{equation}
\label{nuisance3}
    \mathcal{S} = \{(\kappa(\bm{x}), \bm{\xi}) |  \sigma(\kappa(\bm{x}),\bm{\xi}) = 0\},
\end{equation}
where sliding manifold $\sigma(\kappa(\bm{x}),\bm{\xi})$ is
\begin{equation}
\label{nuisance4}
    \sigma(\kappa(\bm{x}),\bm{\xi}) = \xi_2 + \mu \xi_1,
\end{equation}
where $\bm{x} = [\phi \hspace{0.15cm} \theta \hspace{0.15cm} V_T]^\mathrm{T}$ and $\mu = k_1/k_2$ is the adaptive design parameter. The given definition of the adaptive sliding manifold implies that the sliding manifolds are principally functions of the gain $\kappa$, and the gain $\kappa$ is a function of the governing flight states, i.e. $\phi$, $\theta$, and $V_T$. Consequently, the adaptivity of the gain $\kappa$ with respect to $\bm{x} = [\phi \hspace{0.15cm} \theta \hspace{0.15cm} V_T]^\mathrm{T}$ and $\mu = k_1/k_2$ directly leads to the adaptivity of sliding manifolds. For this purpose, the adaptive design of gain $\kappa$ is performed through an optimization framework, which is presented in the proceeding section.

\subsubsection{Optimization Framework}
\label{optimizationFramework}

An optimization problem is set up to calculate the gain pair $k_1$ and $k_2$ in accordance with the timely and aggressive intervention requirements. To obtain the reference command shape, represented in Fig.~\ref{fig:nuisanceVSnuisanceFree}, the principal objectives are determined as follows: \textbf{(1)} the integral of the pitch rate command within its applied time interval, denoted as $[t_0 \hspace{0.15cm} t_f]$, \textbf{(2)} the maximum amplitude of the pitch rate command within its applied time interval, denoted as $[t_0 \hspace{0.15cm} t_f]$, and \textbf{(3)} the norm of the difference between the minimum altitude of the aircraft and the DTED with buffer height, i.e., $h_\text{DTED} + h_\text{buff}$. Thereby, the optimization problem can be described as a multi-objective optimization problem, where the compact objective function is given by Eq.~\eqref{opt1}.

\begin{equation}
\label{opt1}
\begin{aligned}
    J_T &= W_1 \underbrace{\Big[-\int^{t_f}_{t_0}q_\text{GCAS}(t)dt\Big]}_{\substack{J_1}} + W_2 \underbrace{\big[-\max(q_\text{GCAS}(t))\big]}_{\substack{J_2}} + W_3 \underbrace{|\min \big(h(t)\big) - (h_\text{buff} + h_\text{DTED})|}_{\substack{J_3}},
\end{aligned}
\end{equation}
where $W \overset{\Delta}{=} \{W_1, \hspace{0.15cm} W_2, \hspace{0.15cm} W_3\}$ represents the positive weight coefficients corresponding to each objective element. The first element, $J_1$, ensures the continuity of the Auto-GCAS command, eliminating possible discrete commands that could be regarded as nuisance, as shown in Fig.~\ref{fig:nuisanceVSnuisanceFree}. The second element, $J_2$, enhances the aggressiveness of the intervention, while the third element ensures timeliness, preventing interventions that are too early or too late. Consequently, minimizing the compact objective function is expected to produce a timely and aggressive command that is not regarded as a nuisance.

\begin{remark}
    If Proposition~\ref{prop1} and \ref{prop2} are substituted, the optimization problem is reduced to a single variable multi-objective optimization, where $k_1 = k_2^2/4$ and $k_1 \in [4, \infty)$.
\end{remark}

Finally, the optimization set-up is described as given by Eq.~\eqref{opt2}.

\begin{equation}
\label{opt2}
    \begin{aligned}
    k^\star_1,k^\star_2 = \operatorname*{argmin}_{k_1, k_2 \in \mathbb{R}_{>0}} \quad & \sum^{3}_{i = 1} W_i \hat{J}_i \\
    \textrm{s.t.} \quad & k_1 - k_2^2/4 = 0 \\
                  \quad & k_1 \geq 4.
    \end{aligned}
\end{equation}
where $\hat{J}$ represent the normalized objective function between $0$ and $1$.

\section{Flight Envelope Protection System Design}
\label{flightEnvelopeProtection}

The flight envelope protection serves as an additional safety layer to ensure “do no harm" principle. For this purpose, the protection algorithms for the angle of attack, load factor, and bank angle are designed using CBF.

\subsection{Angle of Attack/Load Factor Protection}
\label{AoALoadFactor}

The angle of attack and load factor protection strategy modifies the pitch rate command from Auto-GCAS, generating a safe pitch rate command to ensure compliance with the maximum allowable limits. Then, design a barrier function as presented in Eq.~\eqref{barrierAlpha1}.

\begin{equation}
\label{barrierAlpha1}
    b(\alpha) = \alpha_\text{limit} - \alpha,
\end{equation}
where $\alpha_\text{limit} \in \mathbb{R}$. It is obvious that $b(\alpha) > 0, \forall \alpha \in \mathbb{R}_{< \alpha_\text{limit}}$, and $b(\alpha) = 0 \iff \alpha = \alpha_\text{limit}$. Thus, the time derivative of the barrier function is $\dot{b}(\alpha) = -\dot{\alpha}$. For convenience, the load factor limit can be expressed as an angle of attack limit (due to the bank-to-level requirement, i.e., $\phi = 0^\circ$), as presented in Eq.~\eqref{barrierAlpha2}.

\begin{equation}
\label{barrierAlpha2}
    \alpha^{n_z}_\text{limit} = \dfrac{n_{z_\text{limit}} mg}{\bar{q}_\infty S C_{z_\alpha}},
\end{equation}
where $C_{z_\alpha}$ denotes the z-axis force coefficient derivative with respect to the angle of attack. To ensure safety, the equivalent angle of attack constraint is selected as the more restrictive value between the nominal stall angle $\alpha_\text{limit}$ and the load factor-based limit $\alpha^{n_z}_\text{limit}$. The primary limit for the angle of attack is considered to be the stall angle of attack, $\alpha_\text{stall}$; therefore, the equivalent angle of attack limit can be expressed as given in Eq.~\eqref{barrierAlpha3}.

\begin{equation}
\label{barrierAlpha3}
    \alpha_\text{limit} = \min (\alpha_\text{stall}, \alpha^{n_z}_\text{limit}).
\end{equation}

Remember that the angle of attack dynamics ($\dot{\alpha}$) and the components ($f(\alpha)$ and $g(\alpha)$) that reveal the pitch rate are presented in Eq.~\eqref{barrierH11} and Eq.~\eqref{barrierH13}, respectively. Thus, the CBF constraint for angle of attack and load factor protection is straightforwardly presented in Eq.~\eqref{barrierAlpha4}.

\begin{equation}
\label{barrierAlpha4}
    \underbrace{-f(\alpha) - g(\alpha)q}_{\substack{\mathcal{L}_f b(\alpha) + \mathcal{L}_g b(\alpha) q_\text{cmd}}} + \gamma_\alpha b(\alpha) \geq 0,
\end{equation}
where $\gamma_\alpha \in \mathbb{R}_{>0}$ is the design parameter to be chosen properly. As a consequence, the final form of the angle of attack and load factor protection formulation for generating a pitch rate command is presented in Eq.~\eqref{barrierAlpha5}.

\begin{equation}
\label{barrierAlpha5}
    \begin{aligned}
    u^\star = \operatorname*{argmin}_{q_\text{cmd} \in \mathbb{R}} \quad & \dfrac{1}{2}\big(q_\text{cmd} - q_\text{GCAS}\big)^2 \\
    \textrm{s.t.} \quad & -f(\alpha) - g(\alpha) q_\text{cmd} + \gamma_\alpha b(\alpha) \geq 0 \\
                \quad & q_\text{min} \leq q_\text{cmd} \leq q_\text{max}. \\
    \end{aligned}
\end{equation}

This formulation enables supervision of the Auto-GCAS pitch rate commands to ensure that the allowable safe limits are not violated and generates a safe pitch rate command to the aircraft.

\subsection{Bank Angle Protection}
\label{bankAngle}

The bank angle protection serves to perform the bank-to-level maneuver during recovery; therefore, it can also be regarded as part of the Auto-GCAS structure, with the main strategy being to redesign the pilot's roll rate command, $p_\text{pilot}$. Moreover, it is activated only when $q_\text{pilot} - q_\text{GCAS} \neq 0$, meaning that if Auto-GCAS generates a pitch rate different from the pilot's command, the bank-to-level maneuver is conducted simultaneously. A barrier function is then designed as presented in Eq.~\eqref{barrierPhi1}. 

\begin{equation}
\label{barrierPhi1}
    b(\phi) = \phi_\text{limit} - |\phi|,
\end{equation}
where $\phi_\text{limit} \in \mathbb{R}$. It is obvious that $b(\phi) > 0, \forall \phi \in \mathbb{R}_{< \phi_\text{limit}}$, and $b(\phi) = 0 \iff \phi = \phi_\text{limit}$. Thus, the time derivative of the barrier function is $\dot{b}(\phi) = -\text{sgn}(\phi)\dot{\phi}$. Additionally, since the main function of the bank angle protection is to conduct a bank-to-level maneuver, the limit of the bank angle, $\phi_\text{limit}$, is simply set to $0^\circ$.

At this point, it is necessary to decompose the bank angle dynamics, introduced in Eq.~\eqref{flightDynamicsEq}, into the form $\dot{\phi} = f(\phi) + g(\phi)p$ to reveal the roll rate. Since the roll rate appears explicitly, it is straightforward to describe the necessary components $f(\phi)$ and $g(\phi)$, as presented in Eq.~\eqref{barrierPhi2}.

\begin{equation}
\label{barrierPhi2}
\begin{aligned}
    &f(\phi) = t_\theta(qs_\phi + rc_\phi), \\
    &g(\phi) = 1.
\end{aligned}
\end{equation}

Thus, the CBF constraint for the bank angle protection is straightforwardly presented in Eq.~\eqref{barrierPhi3}. 

\begin{equation}
\label{barrierPhi3}
    \underbrace{-\text{sgn}(\phi)[f(\phi) + g(\phi)p]}_{\substack{\mathcal{L}_f b(\phi) + \mathcal{L}_g b(\phi) p_\text{GCAS}}} + \gamma_\phi b(\phi) \geq 0,
\end{equation}
where $\gamma_\phi \in \mathbb{R}_{>0}$ is the design parameter to be chosen properly. Furthermore, since $\text{sgn}(\phi)$ is a discontinuous function, an approximation of $\text{sgn}(\phi) \approx \tanh\Big(\dfrac{\phi}{\epsilon}\Big)$, where $\epsilon \in \mathbb{R}_{>0}$ is a negligible constant, can be evaluated to eliminate the discontinuity. As a consequence, the final form of the bank angle protection formulation for generating a roll rate command is presented in Eq.~\eqref{barrierPhi4}.

\begin{equation}
\label{barrierPhi4}
    \begin{aligned}
    u^\star = \operatorname*{argmin}_{p_\text{GCAS} \in \mathbb{R}} \quad & \dfrac{1}{2}\big(p_\text{GCAS} - p_\text{pilot}\big)^2 \\
    \textrm{s.t.} \quad & -\tanh\Big(\dfrac{\phi}{\epsilon}\Big)[f(\phi) + g(\phi) p_\text{GCAS}] + \gamma_\phi b(\phi) \geq 0 \\
                \quad & p_\text{min} \leq p_\text{GCAS} \leq p_\text{max}. \\
    \end{aligned}
\end{equation}

The constructed formulation enables the Auto-GCAS to generate a roll rate command ($p_\text{GCAS}$) that closely follows the roll rate command of pilot ($p_\text{pilot}$) while adhering to the constraint for ground collision avoidance, i.e. $\phi = 0^\circ$. Additionally, the generated roll rate command, $p_\text{GCAS}$, must remain within the interval $[p_\text{min} \hspace{0.15cm} p_\text{max}]$, considering the admissible and allowable roll rate limits depending on the aircraft. Since there does not exist an additional layer for the bank angle protection, the command $p_\text{GCAS}$ is equivalent to the command $p_\text{cmd}$. 

\begin{remark}
    It is noteworthy that the bank angle must be within the interval, $-\pi \leq \phi \leq \pi$, by definition.
\end{remark}

The overall proposed architecture is depicted in Fig.~\ref{generalFramework} for clarity.

\begin{figure*}[hbt!]
\centering
\includegraphics[width=\textwidth]{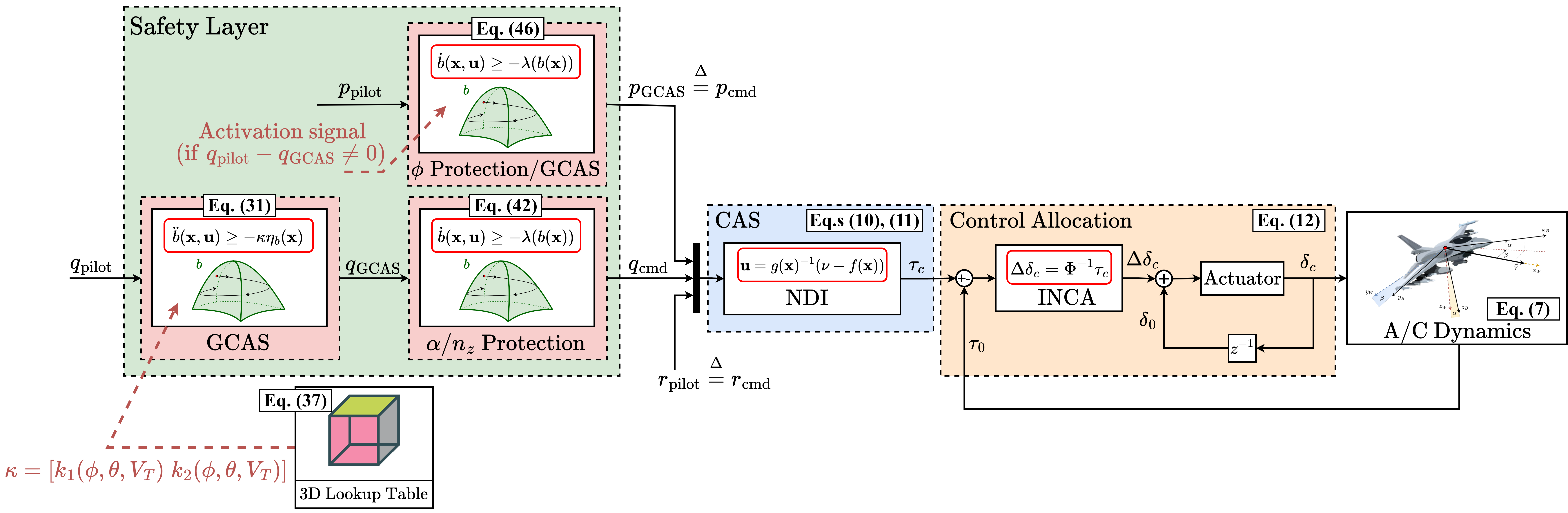}
\caption{General framework of the proposed method: \textbf{(1)} Safety layer including GCAS and FEP, \textbf{(2)} CAS including NDI, \textbf{(3)} Control allocation, and \textbf{(4)} A/C dynamics.}
\label{generalFramework}
\end{figure*}

The pilot commands \( p_\text{pilot} \), \( q_\text{pilot} \), and \( r_\text{pilot} \), with roll and pitch rate commands subject to GCAS constraints. The pitch rate command is adjusted by the ECBF constraint, and \( q_\text{GCAS} \) is generated if needed. This command is also limited by angle of attack and load factor protection, generating \( q_\text{cmd} \) when required. Bank angle protection initiates a bank-to-level maneuver, generating \( p_\text{GCAS} \) if \( q_\text{pilot} - q_\text{GCAS} \neq 0 \), indicating a recovery maneuver initiated; otherwise, the pilot's roll rate command, \( p_\text{pilot} \), is passed through.

\section{Results}
\label{results}

In order to calculate the adaptive ECBF gain vector, the optimization is performed under various conditions by meshing $\phi$, $\theta$, and $V_T$ within the intervals $[-150^\circ \hspace{0.15cm} 150^\circ]$ with an increment of $\Delta\phi = 50^\circ$, $[-60^\circ \hspace{0.15cm} -10^\circ]$ with an increment of $\Delta\theta = 10^\circ$, and $[200 \hspace{0.15cm} 350]$ with an increment of $\Delta V_T = 30$m/s, respectively. The other flight states are randomly assigned to ensure the aircraft dives. Finally, the optimization is carried out at each design point, assuming the ground is flat and at zero level, i.e., $h_\text{DTED} = 0$m while keeping $h_\text{buff} = 100$~m, and the calculated gains $k_1$ and $k_2$ are scheduled as functions of $\phi$, $\theta$, and $V_T$. Moreover, the proposed Auto-GCAS design has been evaluated through three different assessments, \textbf{(1)} terrain collision avoidance scenarios, \textbf{(2)} an authority-sharing scenario, and \textbf{(3)} Monte Carlo simulations. Note that, for all assessments, the buffer height is set to 100 meters, and the center of gravity of the aircraft is set to $35\%\bar{c}$, which makes the aircraft statically unstable in the pitch axis.

\subsection{Terrain Collision Avoidance Assessments}
\label{terrainCollisionResults}

Two distinct terrain collision avoidance scenarios are conceptualized: one at low altitude with randomized state variables and another at high altitude with randomized states variables. The term “randomized state variables" encompasses aerodynamic angles, Euler angles, angular rates, and true velocity, with altitude being explicitly excluded. A default terrain model, synthesized using MATLAB\textsuperscript{\tiny\textregistered}'s \textit{peaks} function, is used to generate a smooth, synthetic terrain surface for visualization and conceptual validation of the barrier formulation; it does not represent real-world terrain. Furthermore, it is assumed that the pilot is incapacitated (e.g. experiencing blackout), resulting in the absence of control inputs, such that angular rate commands are maintained at $0^\circ/\text{s}$. Finally, a terrain scanning pattern should be designed to account for non-flat, mountainous terrain. However, since developing such a method is outside the scope of this study, a default rectangular scan pattern is employed, extending 750 meters ahead of the aircraft and 150 meters to each side (left and right).

\subsubsection{Scenario-1: Low Altitude}

In this scenario, the aircraft begins its flight with randomized state variables at an altitude of 1000 meters. The relevant state trajectories are depicted in Fig.~\ref{statesTerrainMan1}. 

\begin{figure}[hbt!]
\centering
\includegraphics[width=0.9\textwidth]{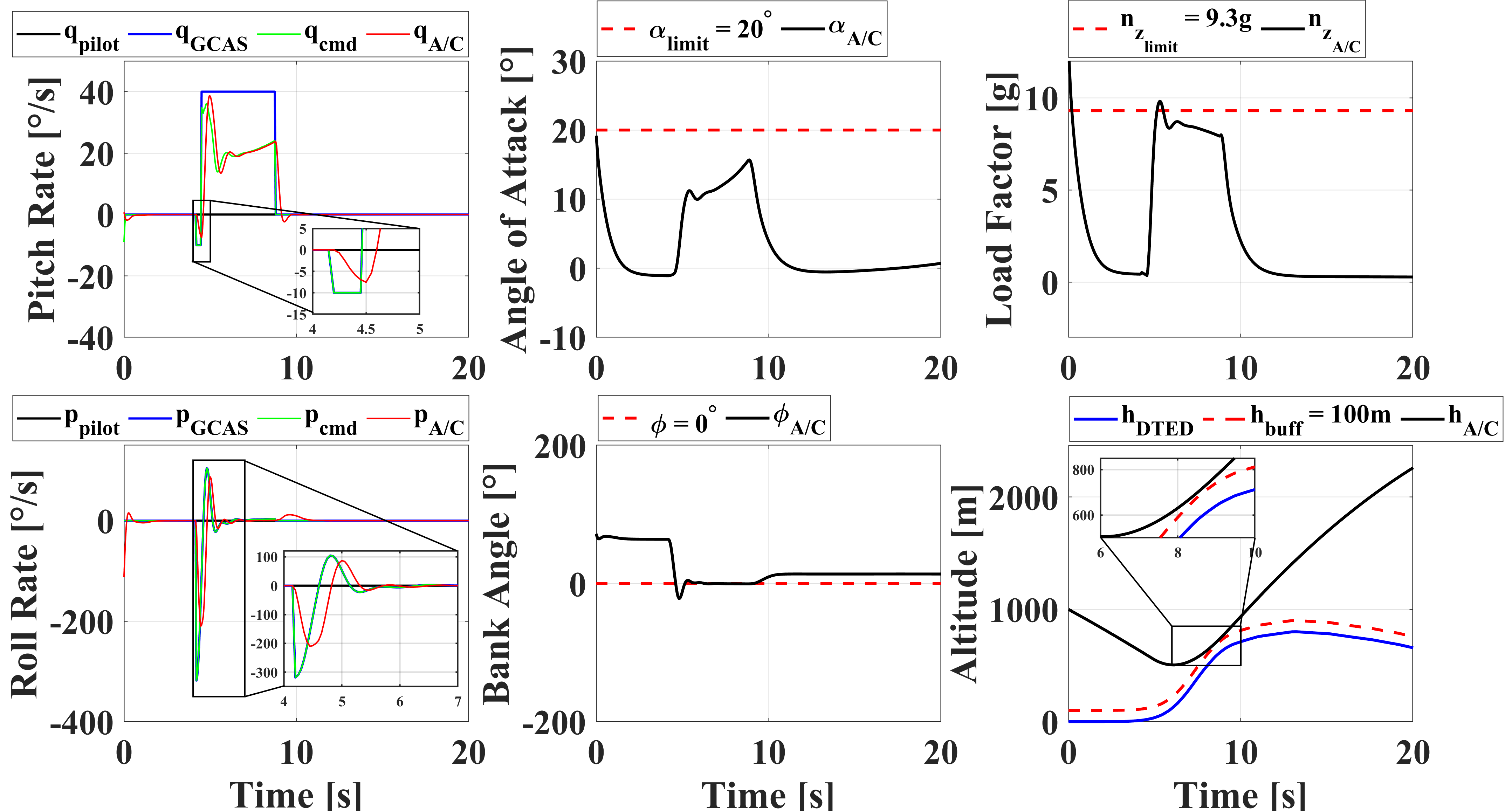}
\caption{Corresponding state trajectories: \textbf{“pilot”} = pilot input, \textbf{“GCAS”} = Auto-GCAS command, \textbf{“cmd”} = resultant command, \textbf{“A/C”} = aircraft response.}
\label{statesTerrainMan1}
\end{figure}

Given the randomized initial conditions, the aircraft's states at the start are consistent with these specifications. Subsequently, the aircraft stabilizes itself under equilibrium conditions, driven by the $0^\circ/\text{s}$ angular rate commands. For a certain period, no intervention from the Auto-GCAS is observed, indicating that it waits for the most appropriate moment to act. Following the fourth second of the dive, the Auto-GCAS engages, commanding the aircraft to bank-to-level and execute a pull-up maneuver. The Auto-GCAS commands are continuous and at maximum amplitude. However, under the supervision of the FEP algorithm, these commands are subjected to flight envelope constraints, specifically angle of attack limit $\alpha_\text{limit}$, and load factor limit $n_{z_\text{limit}}$. To prevent exceeding the allowable load factor, the commanded pitch rate is reduced accordingly. Additionally, it is observed that the aircraft accurately tracks the adjusted pitch rate command. At the onset of Auto-GCAS intervention, both pitch and roll maneuvers are initiated. This results in a bank-to-level maneuver, aligning the bank angle $\phi$ to $0^\circ$. By the conclusion of the recovery maneuver, the aircraft clears the buffer height, demonstrating that the recovery commanded by Auto-GCAS was both timely and aggressive. Once the recovery is successfully achieved, Auto-GCAS ceases its commands, and the angular rate commands return to their original values $0^\circ/\text{s}$. The corresponding control surface deflections and control allocation objective histories are presented in Fig.~\ref{controlDeflectionsTerrainMan1}.

\begin{figure}[hbt!]
\centering
\includegraphics[width=0.9\textwidth]{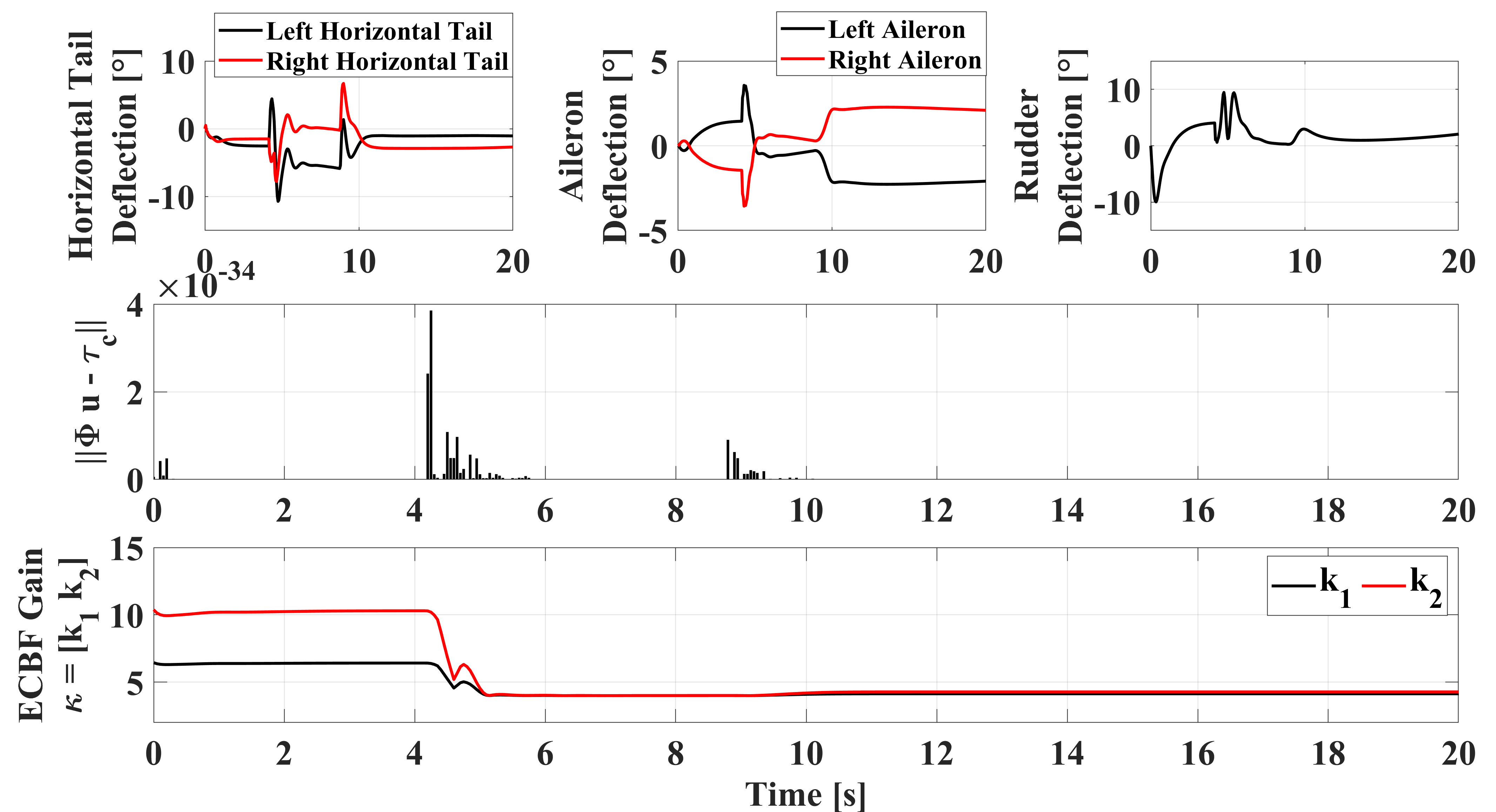}
\caption{Corresponding control surface deflections, control allocation objective, $||\phi u - \tau_c||$, and adaptive ECBF control gain vector.}
\label{controlDeflectionsTerrainMan1}
\end{figure}

The most significant outcome is the satisfaction of the control moment coefficients through allocation on the control surfaces, meaning that the commanded control moment coefficients have been achieved. Finally, a 3D illustration of the scenario is depicted in Fig.~\ref{3DTerrainMan1}.

\begin{figure}[hbt!]
\centering
\includegraphics[width=0.9\textwidth]{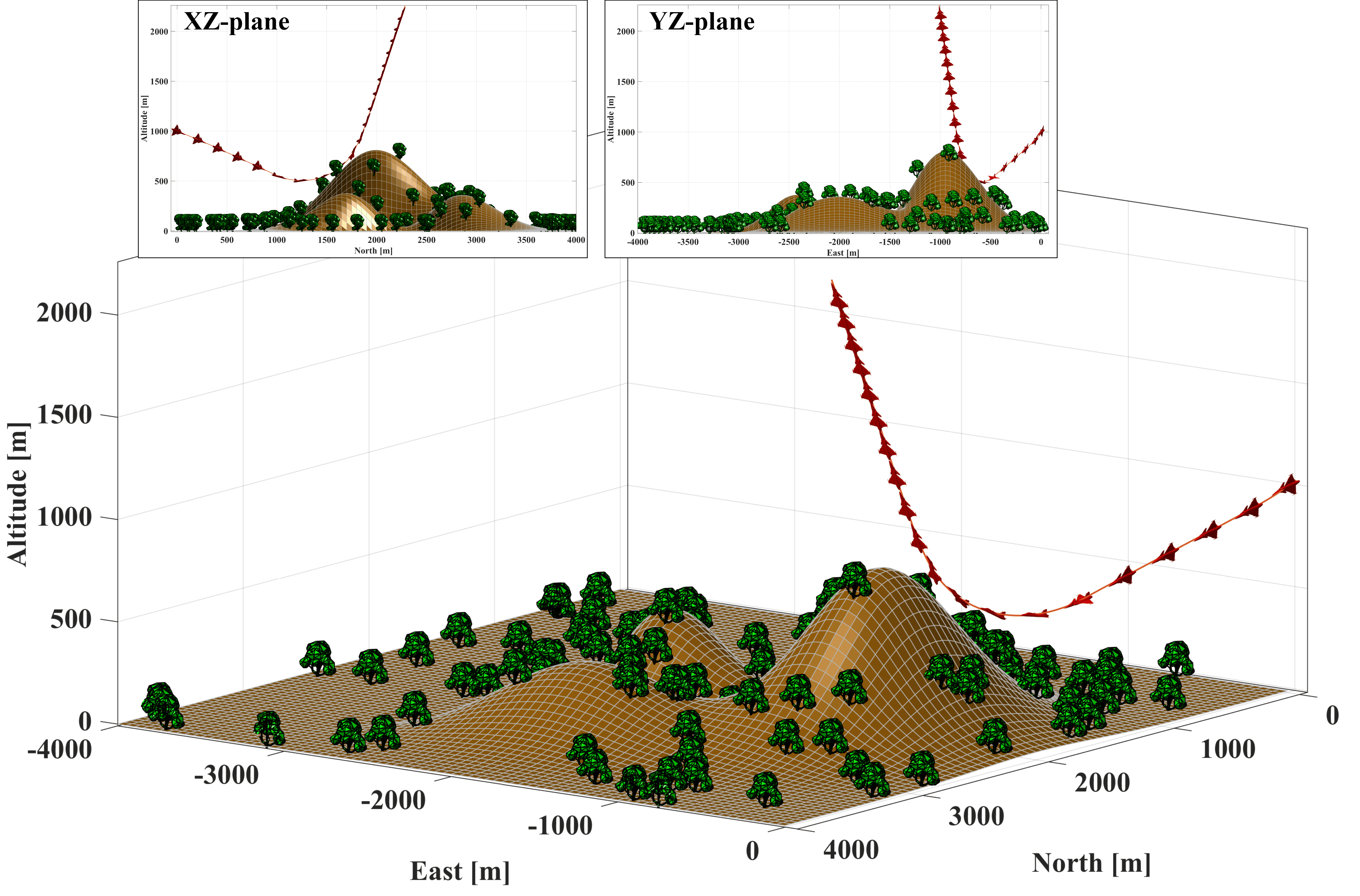}
\caption{3D illustration of Scenario-1: from XZ plane, YZ plane, and isometric view.}
\label{3DTerrainMan1}
\end{figure}

\subsubsection{Scenario-2: High Altitude}

In this scenario, the aircraft begins its flight with randomized state variables at an altitude of 3500 meters. The relevant state trajectories are depicted in Fig.~\ref{statesTerrainMan2}. 

\begin{figure}[hbt!]
\centering
\includegraphics[width=0.9\textwidth]{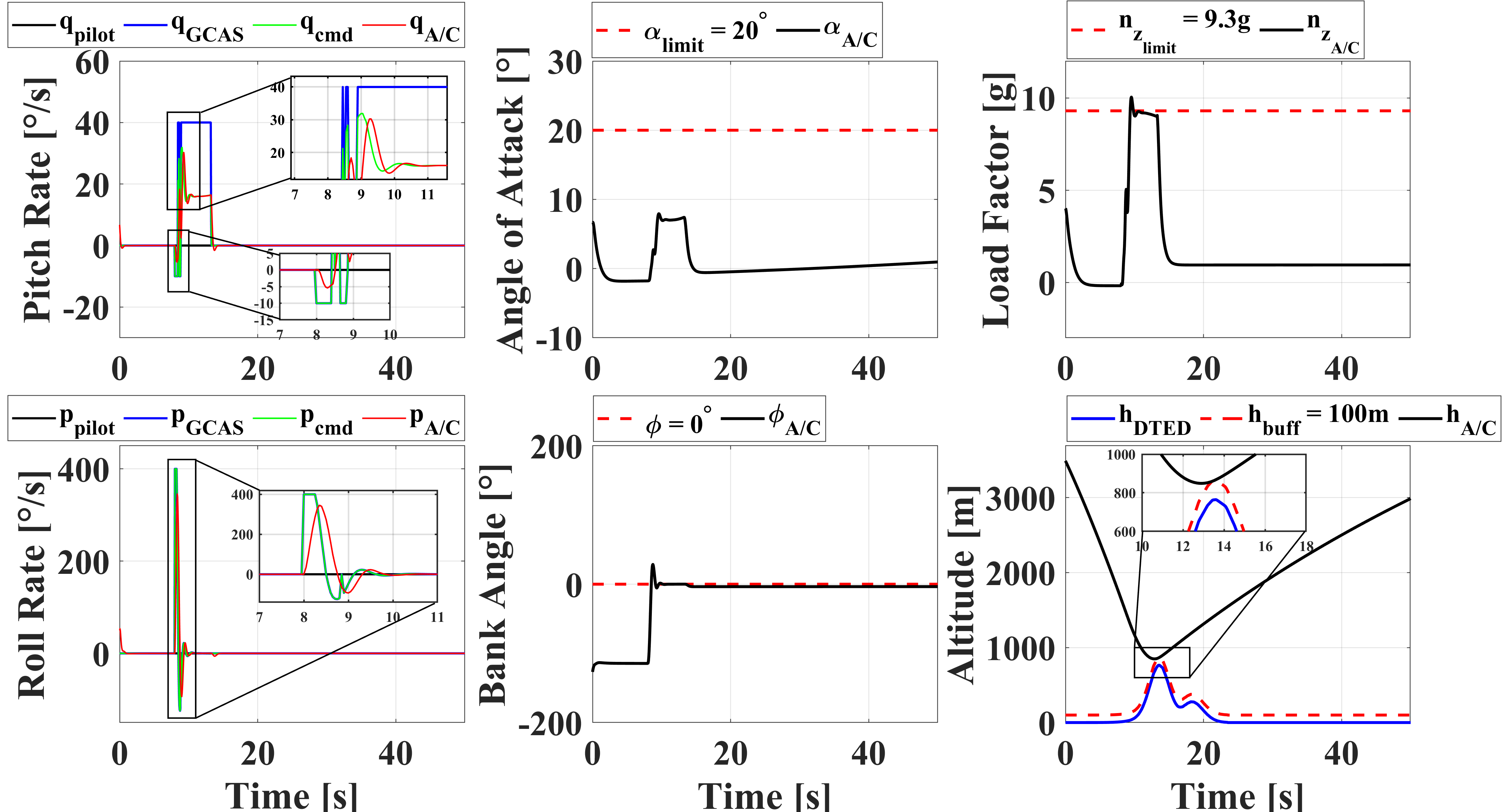}
\caption{Corresponding state trajectories: \textbf{“pilot”} = pilot input, \textbf{“GCAS”} = Auto-GCAS command, \textbf{“cmd”} = resultant command, \textbf{“A/C”} = aircraft response.}
\label{statesTerrainMan2}
\end{figure}

Once again, for a certain amount of time, the intervention of Auto-GCAS is not observable. Approximately at the \nth{8} second of the dive, Auto-GCAS intervenes, commanding the aircraft to bank-to-level and execute a pull-up maneuver. The Auto-GCAS pitch rate command is nearly continuous and at maximum amplitude. However, such a high-amplitude command is not permitted by FEP due to the angle of attack limit, $\alpha_\text{limit}$, and the load factor limit, $n_{z_\text{limit}}$. It is observable that, for a certain period, the aircraft achieves maximum vertical acceleration, but any excess is not allowed. Due to the simultaneous initiation of both pitching and rolling during the recovery maneuver, the bank angle is rapidly set to $0^\circ$. Seemingly, not only is the pitch rate applied at its maximum amplitude, but also the roll rate. At the end of the recovery maneuver, the aircraft clears the buffer height for this scenario as well, indicating that the commanded recovery maneuver by Auto-GCAS is once again timely and aggressive. Once the recovery is successfully completed, Auto-GCAS ceases commanding, and the commanded angular rates return to the original values, i.e., $0^\circ/\text{s}$. The corresponding control surface deflections and control allocation objective histories are presented in Fig.~\ref{controlDeflectionsTerrainMan2}.

\begin{figure}[hbt!]
\centering
\includegraphics[width=0.9\textwidth]{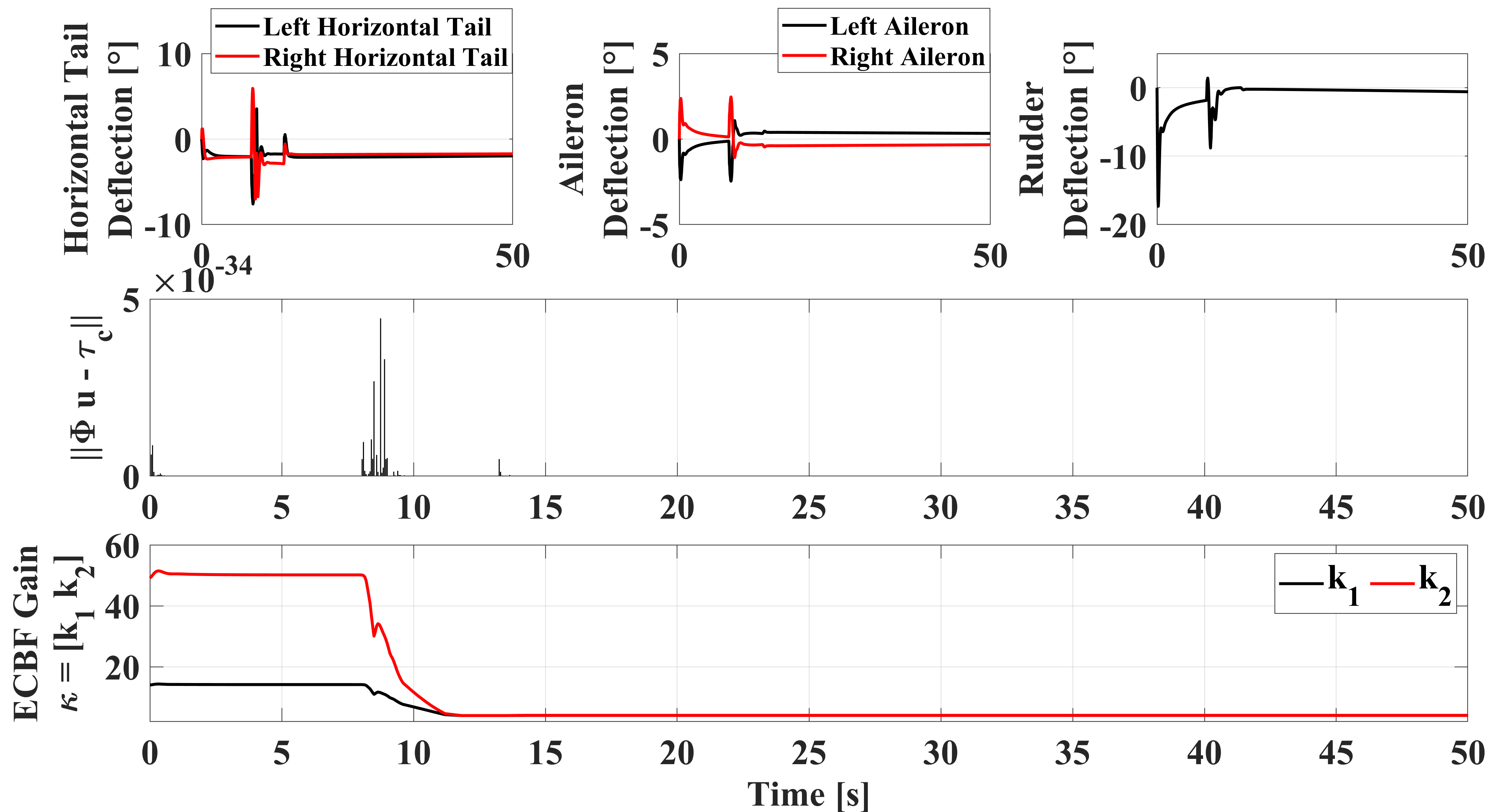}
\caption{Corresponding control surface deflections, control allocation objective, $||\phi u - \tau_c||$, and adaptive ECBF control gain vector}
\label{controlDeflectionsTerrainMan2}
\end{figure}

The history of the control allocation objective indicates that, once again, the commanded control moment coefficients have been achieved. Therefore, the baseline flight control law also works properly and effectively. Finally, a 3D illustration of the scenario is depicted in Fig.~\ref{3DTerrainMan2}.

\begin{figure}[hbt!]
\centering
\includegraphics[width=0.7\textwidth]{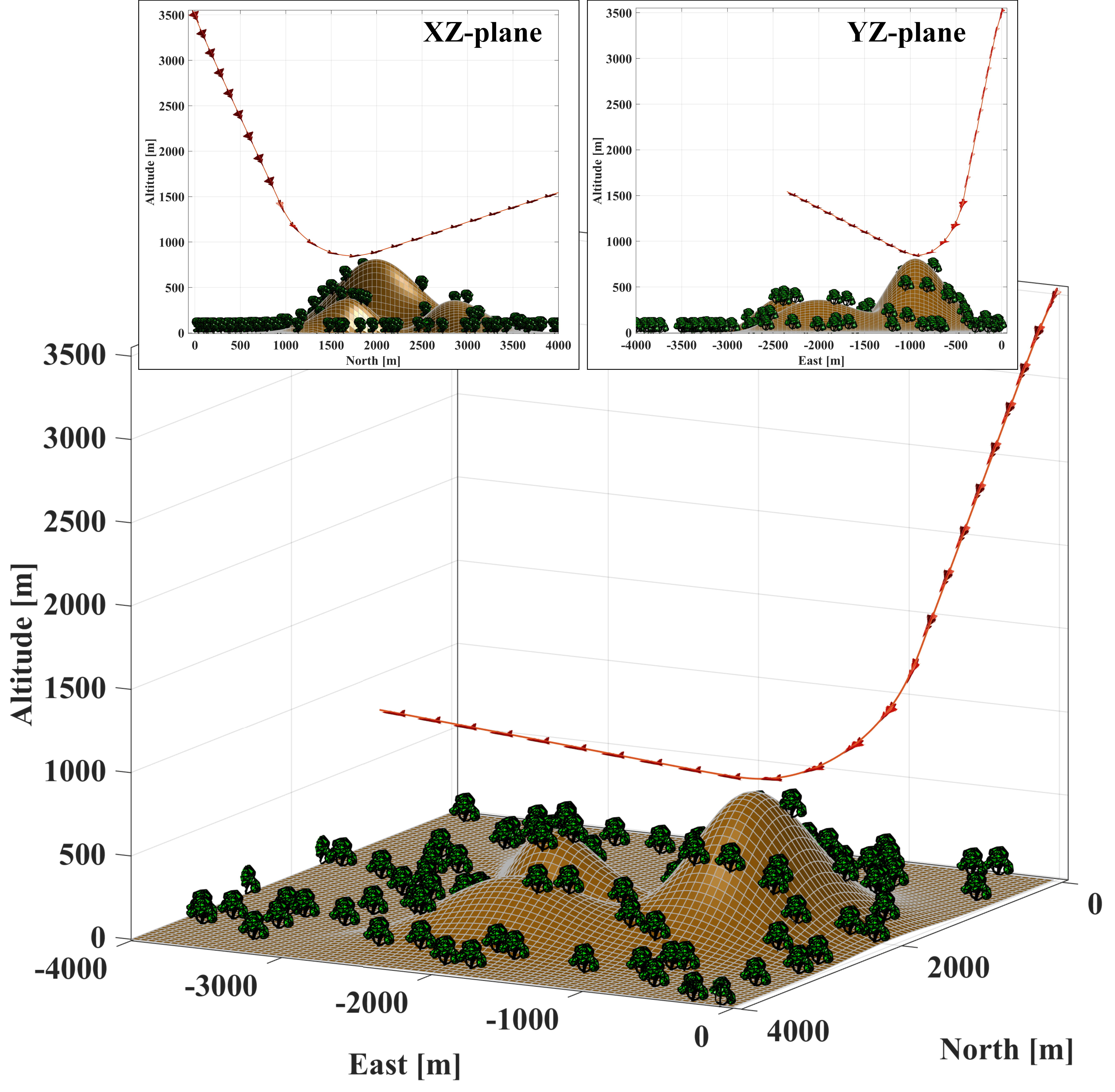}
\caption{3D illustration of Scenario-2: from XZ plane, YZ plane, and isometric view.}
\label{3DTerrainMan2}
\end{figure}

\subsection{Authority-sharing Assessment}
\label{authoritySharingResults}

This assessment should be evaluated within a different context than the previous assessments. In this scenario, the pilot remains capacitated and commands arbitrary angular rates at different amplitudes at various time instants over a certain time period. However, Auto-GCAS remains activated, and the pilot is not allowed to violate the buffer height. The key point to consider is the extent to which Auto-GCAS intervenes with the pilot's commands and when it ceases the intervention. In this regard, the scenario is conceptualized as follows: the aircraft dives under randomized conditions at an altitude of 3500 meters. After losing a certain amount of altitude, the pilot commands rolling at a specific instant, pitch-roll coupling at another specific instant, and only pitching at yet another specific instant. The relevant state trajectories are depicted in Fig.~\ref{authoritySharingStates}.

\begin{figure}[hbt!]
\centering
\includegraphics[width=0.9\textwidth]{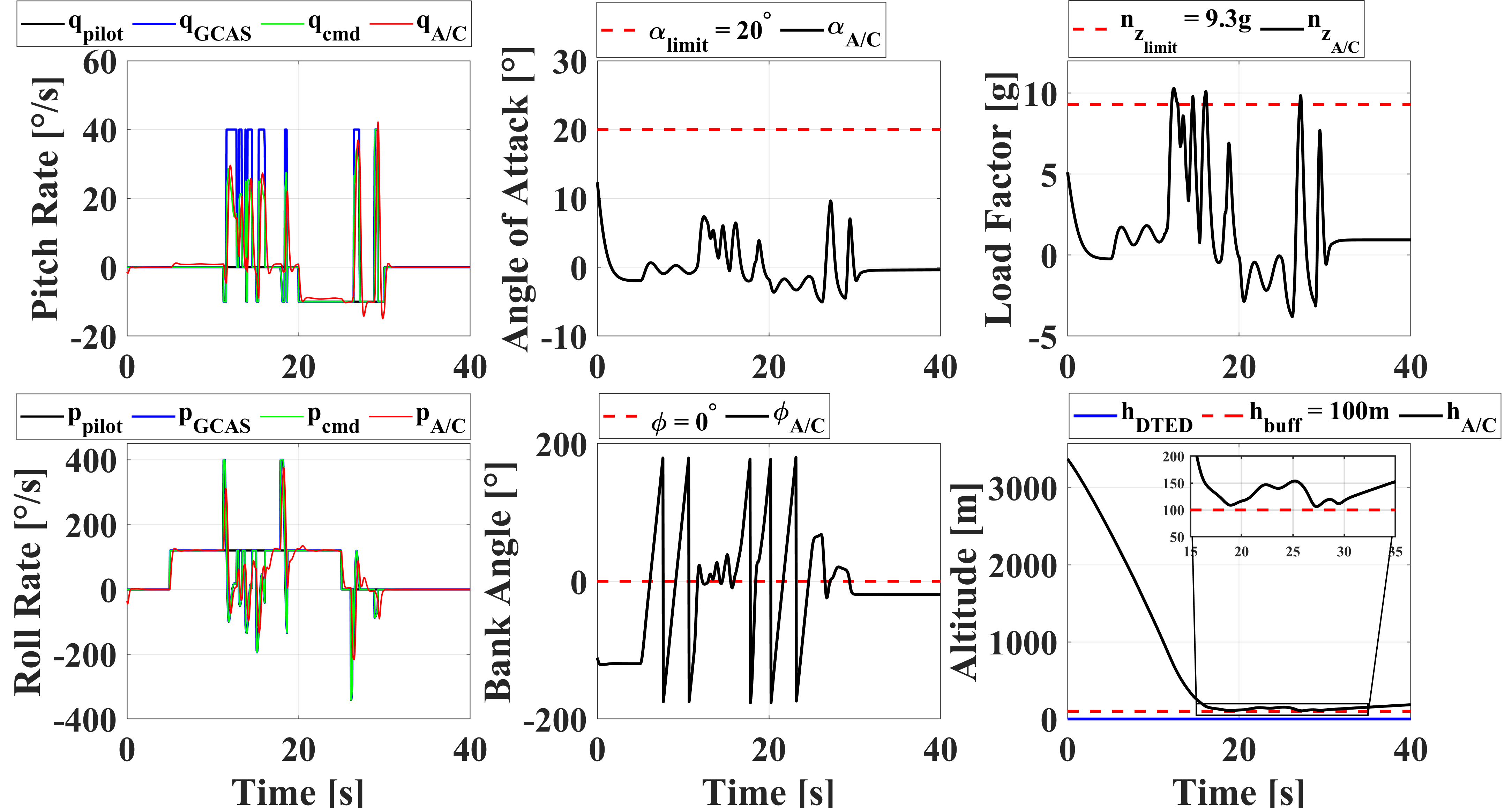}
\caption{Corresponding state trajectories: \textbf{“pilot”} = pilot input, \textbf{“GCAS”} = Auto-GCAS command, \textbf{“cmd”} = resultant command, \textbf{“A/C”} = aircraft response.}
\label{authoritySharingStates}
\end{figure}

The states of this scenario should be interpreted in a different context. Primarily, since the pilot persistently commands the aircraft to dive through the buffer height, the expected nuisance-free intervention cannot be anticipated. Consequently, the intervention commands of Auto-GCAS differ from those depicted in Fig.~\ref{fig:nuisanceVSnuisanceFree}. The initial Auto-GCAS commands, prior to reaching the buffer height, involve pull-up and bank-to-level maneuvers. The pilot's roll command is adjusted to achieve wings-level, and a pitch-up command is initiated simultaneously. However, the other safety layer, i.e. FEP, prevents the load factor from exceeding its limit, thus modifying the Auto-GCAS command. Nevertheless, a significant aspect of this response is the chattering-like behavior in the Auto-GCAS pitch rate commands. It is observable that the aircraft clears the buffer height due to the initial actions of Auto-GCAS. However, the pilot continues to command both push-over and roll. In such a scenario, the primary responsibility of Auto-GCAS is to protect the aircraft from a collision, making the expectation of a nuisance-free intervention contrary to the rationale. Furthermore, if the altitude response of the aircraft is examined in detail, the Auto-GCAS commands are issued at the correct time and at maximum amplitude, indicating a timely and aggressive intervention. As soon as the risk of violating the buffer height is mitigated, Auto-GCAS ceases its intervention, and the pilot's commands are once again permitted. Therefore, a collaboration between the pilot and Auto-GCAS is evident, provided that safety is ensured. The corresponding control surface deflections and control allocation objective histories are presented in Fig.~\ref{controlDeflectionsAuthoritySharing}.

\begin{figure}[hbt!]
\centering
\includegraphics[width=0.9\textwidth]{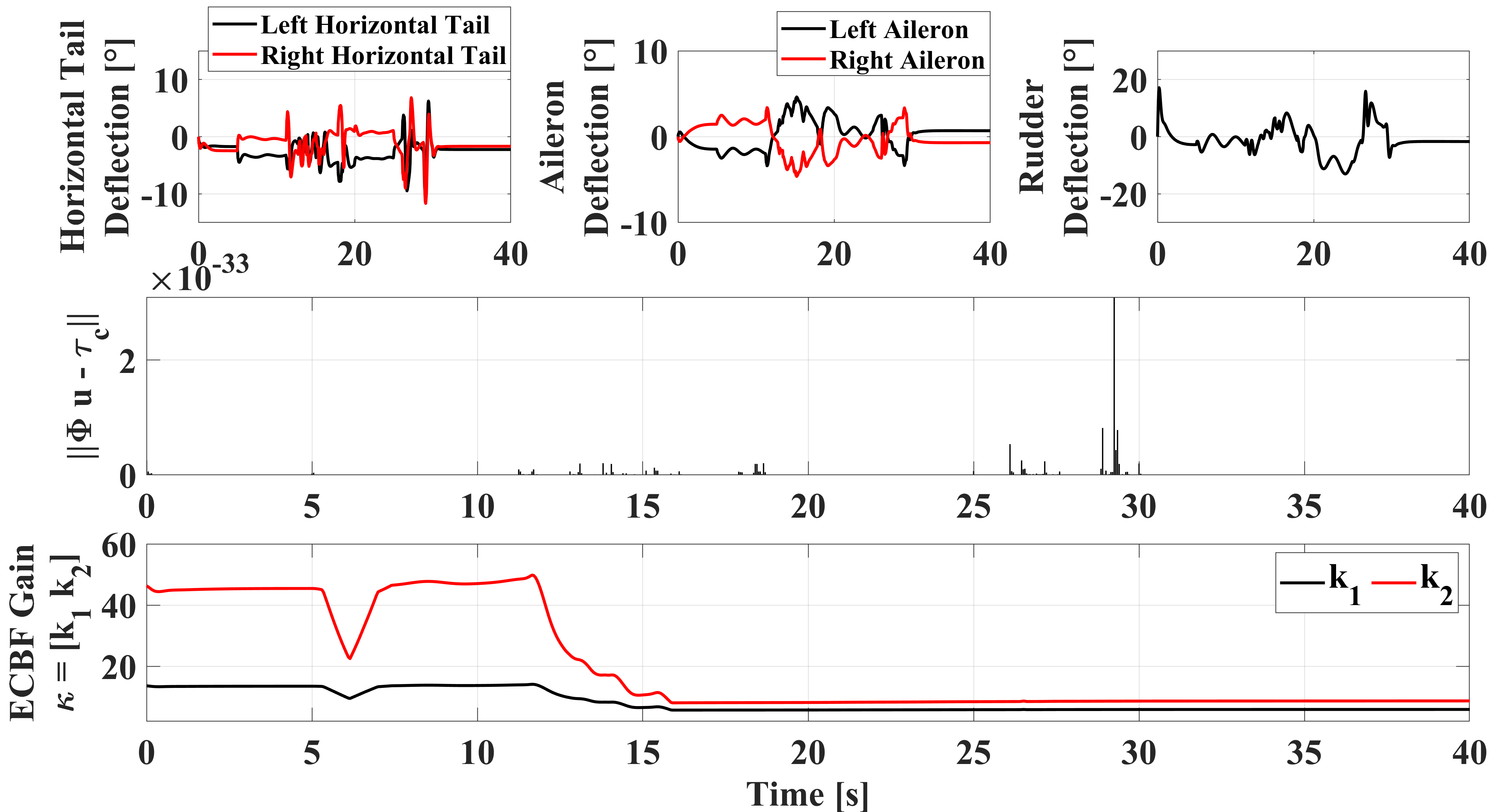}
\caption{Corresponding control surface deflections, control allocation objective, $||\phi u - \tau_c||$, and adaptive ECBF control gain vector}
\label{controlDeflectionsAuthoritySharing}
\end{figure}

The history of the control allocation objective indicates that, once again, the commanded control moment coefficients have been achieved. Therefore, the baseline flight control law also functions properly and effectively for this scenario. Finally, a 3D illustration of the scenario is depicted in Fig.~\ref{3DTerrainAuthoritySharing}.

\begin{figure}[hbt!]
\centering
\includegraphics[width=0.9\textwidth]{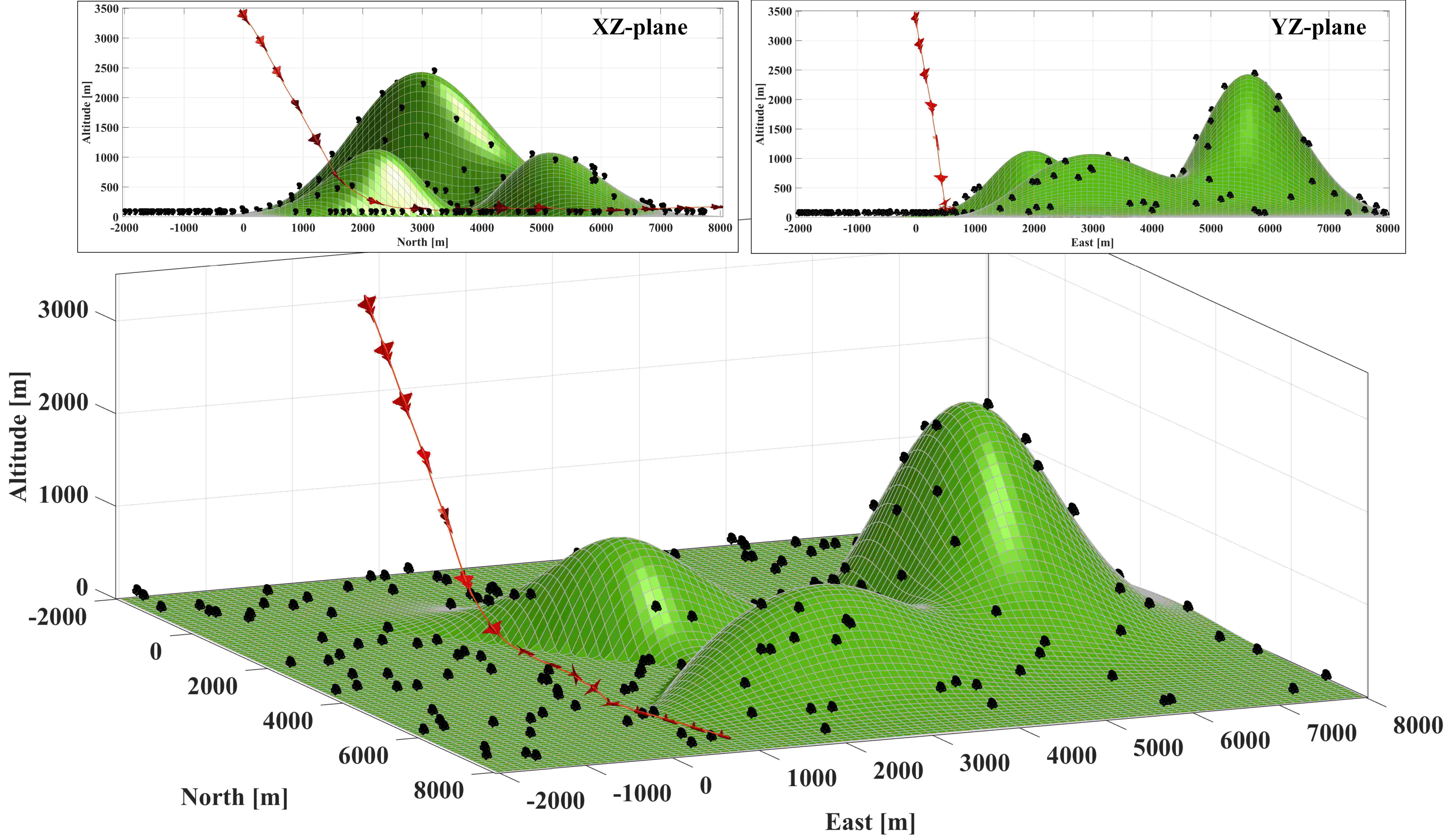}
\caption{3D illustration of Scenario-2: from XZ plane, YZ plane, and isometric view.}
\label{3DTerrainAuthoritySharing}
\end{figure}

\subsection{Monte Carlo Simulations}
\label{monteCarloResults}

Monte Carlo simulations are conducted under $850$ randomized initial conditions to investigate the efficacy of the proposed method. While randomized initial conditions are used to evaluate general robustness, additional tailored scenarios\textemdash such as steep dives and low-altitude banked turns\textemdash has already been presented to challenge the controller under known high-risk conditions. In this assessment, since the DTED is assumed to be flat and at ground level (0 meters), the previously mentioned terrain scan pattern has not been utilized. All the states of the aircraft are randomized, as depicted in Fig.~\ref{initialConditionSpace}.

\begin{figure}[hbt!]
\centering
\includegraphics[width=0.8\textwidth]{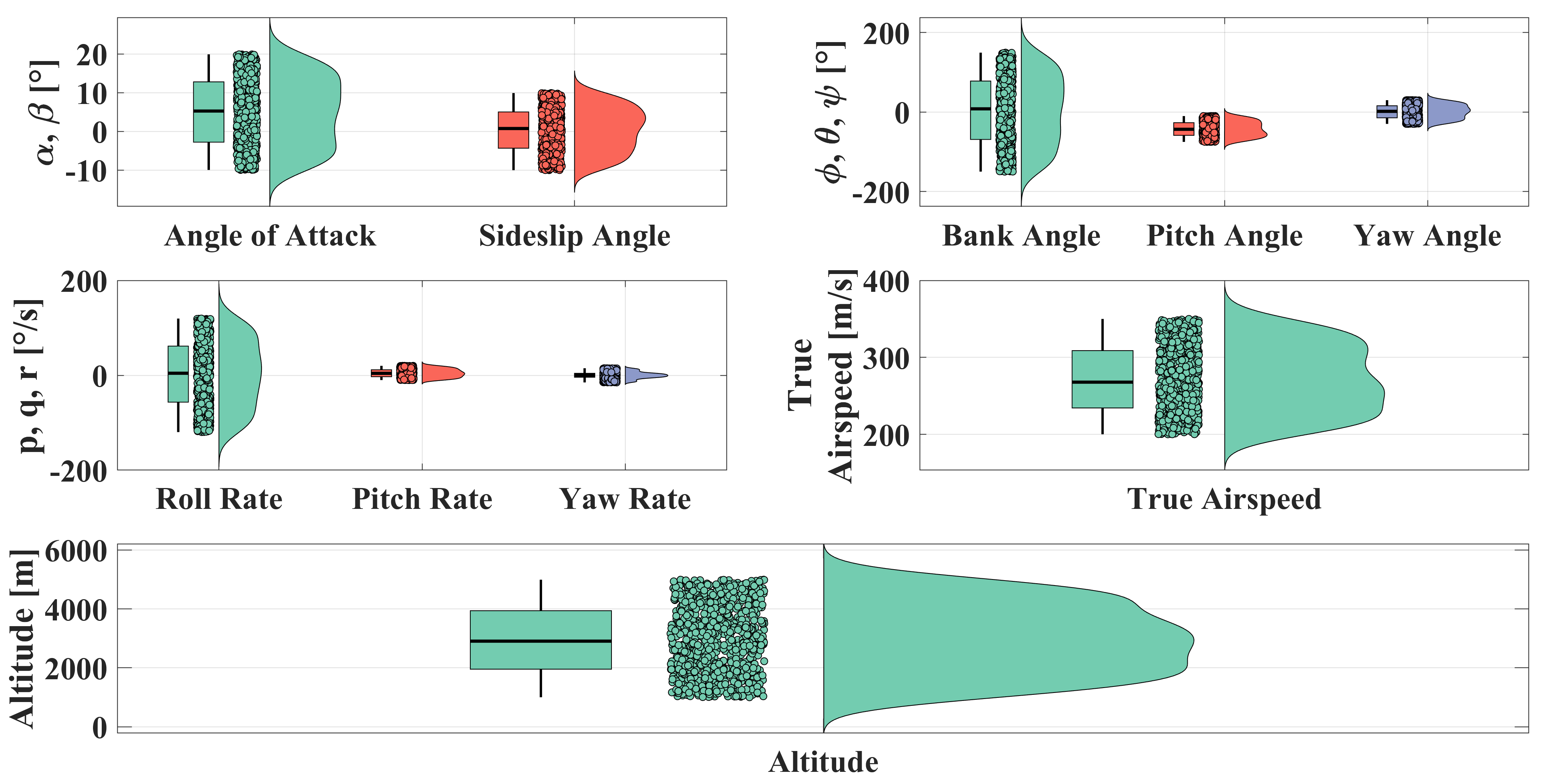}
\caption{Initial condition space for Monte Carlo simulations.}
\label{initialConditionSpace}
\end{figure}

The angle of attack is randomized within the interval \( [-10^\circ, \hspace{0.15cm} 20^\circ] \), sideslip within \( [-10^\circ, \hspace{0.15cm} 10^\circ] \), bank angle within \( [-150^\circ, \hspace{0.15cm} 150^\circ] \), pitch angle within \( [-75^\circ, \hspace{0.15cm} -10^\circ] \), yaw angle within \( [-30^\circ, \hspace{0.15cm} 30^\circ] \), roll rate within \( [-120^\circ/\text{s}, \hspace{0.15cm} 120^\circ/\text{s}] \), pitch rate within \( [-10^\circ/\text{s}, \hspace{0.15cm} 20^\circ/\text{s}] \), yaw rate within \( [-15^\circ/\text{s}, \hspace{0.15cm} 15^\circ/\text{s}] \), true airspeed within \( [200\text{m/s}, \hspace{0.15cm} 350\text{m/s}] \), and altitude within \( [1000\text{m}, \hspace{0.15cm} 5000\text{m}] \). Understandably, the search space is meaningfully broad. The performance metrics for the evaluation are the minimum altitude, maximum pitch rate command, maximum load factor, and the Euclidean distance between the reference command depicted in Fig.~\ref{fig:nuisanceVSnuisanceFree} and the applied pitch rate command. In order to obtain a meaningful Euclidean distance metric, the applied pitch rate and reference commands are normalized between $0$ and $1$. Then, the dynamic time warping method is employed to compare the similarity between the shapes of the signals by aligning them in time to minimize differences, even if they are out of phase or vary in speed. This is achieved using MATLAB\textsuperscript{\tiny\textregistered}'s \textit{dtw} function, which calculates an optimal match between two sequences by warping their time axes to minimize the cumulative distance between corresponding points. Consequently, the comparison is independent of the application time and magnitudes. The statistical illustration of the Monte Carlo simulations are presented in Fig.~\ref{MonteCarloResults}. 

\begin{figure}[hbt!]
\centering
\includegraphics[width=0.8\textwidth]{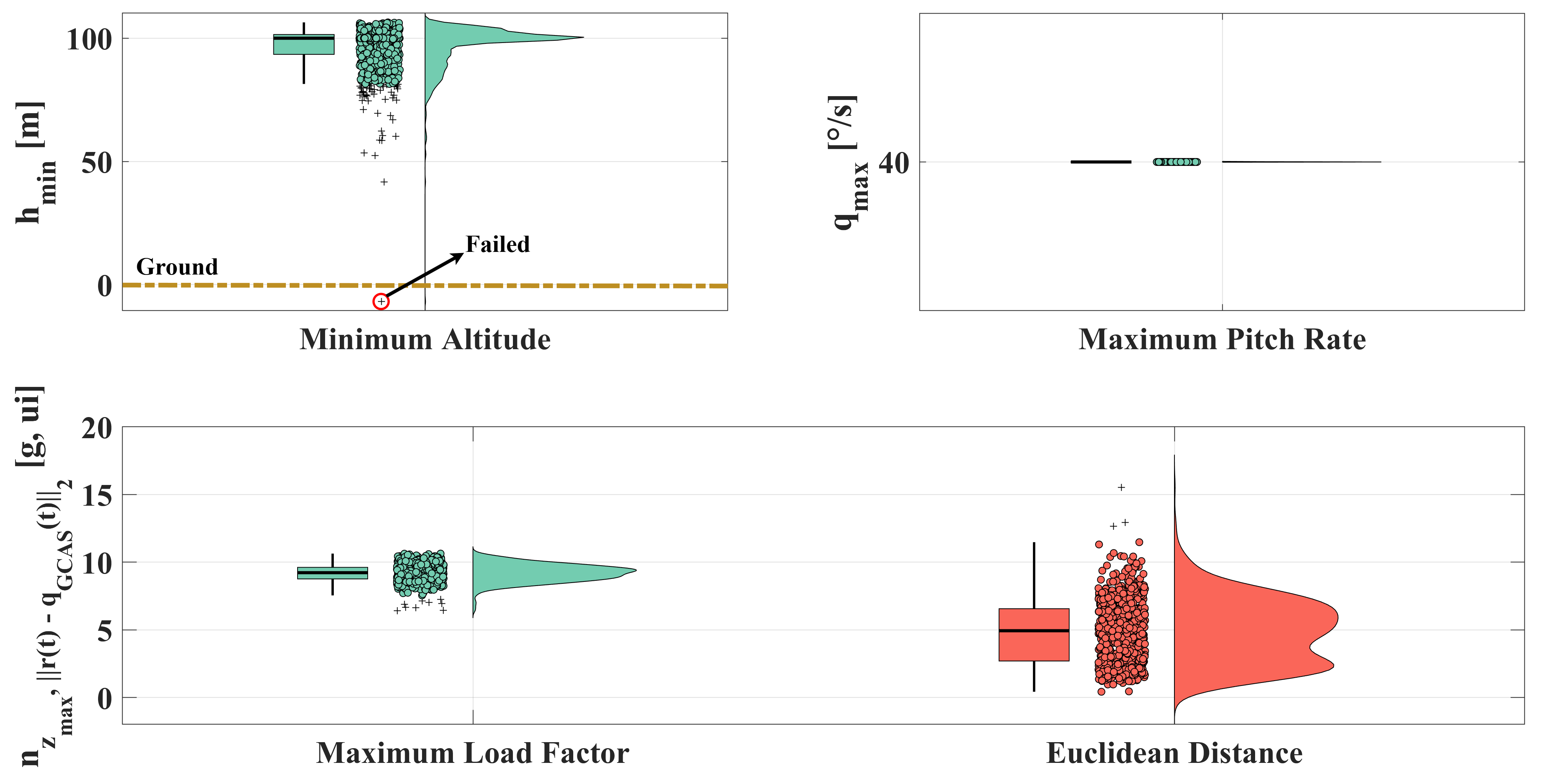}
\caption{Performance metrics of Monte Carlo simulations: minimum altitude, maximum pitch rate, maximum load factor, Euclidean distance between the reference signal and applied pitch rate.}
\label{MonteCarloResults}
\end{figure}

The results indicate a high success rate of the proposed methodology. The average minimum altitude is 100 meters, which corresponds to the predetermined buffer height. This outcome is a clear indication of timely intervention. Additionally, the maximum pitch rate command and maximum load factor reflect the aggressiveness of the Auto-GCAS intervention. The outcomes that should be elaborated are the failed case and the Euclidean distance. The initial conditions for the failed case are \( 1120.6 \) meters for the altitude, \( -64.86^\circ \) for the pitch angle, \( 96.96^\circ \) for the bank angle, and \( 330.1 \, \text{m/s} \) for the true airspeed. These parameters are of concern due to their significant impact in a dive scenario. Notably, the altitude is relatively low to protect the aircraft from such a steep dive (\( -64.86^\circ \) for the pitch angle) and high speed (\( 330.1 \, \text{m/s} \), circa 1 Mach). The histories of this case reveal that the applied pitch rate corresponds to a 9G pull-up at the beginning of the dive and it sustained for approximately 6 seconds; however, this recovery maneuver still fails to prevent a collision. Consequently, this case can be regarded as both extreme and harsh. Furthermore, the Euclidean distance should be as close to \( 0 \) as possible for a perfect overlap between the reference signal and the applied pitch rate command. The average Euclidean distance is approximately \( 5 \). To interpret this value, a comparison is essential. Therefore, the Euclidean distances of the pitch rate commands depicted in Fig.~\ref{statesTerrainMan1} and Fig.~\ref{statesTerrainMan2} are \( 1.2 \) and \( 5.1611 \), respectively, and can be regarded as reference values. Thus, it can be interpreted that the applied pitch rate commands closely resemble the pitch rate shown in Fig.~\ref{statesTerrainMan2}, which is satisfactory in terms of the nuisance-free criterion. Consequently, it is evident that the proposed methodology achieves its objectives with a high success rate of $849$ successful collision avoidance out of $850$ random cases, indicating a \( 99.88\% \) success rate.

\section{Discussion}
\label{discussion}

This section discusses the practical applicability, comparative advantages, aircraft suitability, and limitations of the proposed nuisance-free Auto-GCAS framework, with emphasis on robustness and implementation feasibility.

\subsection{Practical Implementation and Applicability}

The proposed method is designed with real-time implementation in mind. By formulating collision avoidance and flight envelope protection as convex quadratic programs, the algorithm can be integrated into modern flight control computers without requiring extensive computational resources. The adaptive ECBF gains are computed offline across a range of flight conditions (pitch angle, bank angle, airspeed), enabling efficient onboard scheduling without online optimization. Integration into existing avionics architectures is feasible, requiring access to standard state estimates, DTED, and pilot commands, which are common in current modern and advanced aircraft regardless of manned or unmanned. Additionally, the framework is applicable to fixed-wing aircraft platforms, including manned and unmanned aircraft with comparable control authority. It is particularly suitable for avoiding terrain-based obstacles represented with DTED, such as hills or mountainous terrain.

\subsection{Addressing Literature Gaps and Comparative Strengths}

Unlike prior nuisance-free Auto-GCAS approaches that rely on computationally demanding trajectory prediction or simplified 3DoF dynamics, this method uses adaptive ECBFs to generate timely and aggressive recovery maneuvers without sacrificing model fidelity. Additionally, it incorporates a flight envelope protection layer based on control barrier functions, addressing a key shortcoming in earlier studies that lacked comprehensive safety enforcement. Simulation results, including Monte Carlo assessments with a $99.88\%$ success rate, demonstrate the method’s ability to meet all critical objectives: ground collision avoidance, nuisance mitigation, and envelope protection. The avoidance behavior is both continuous and aggressive, closely emulating pilot-preferred interventions and minimizing false alarms.

\subsection{Limitations and Robustness Considerations}

Despite its strengths, the method has some limitations. While the gain schedule improves adaptivity, its resolution may limit performance in rapidly changing flight conditions. This could be addressed by incorporating real-time interpolation or machine learning-based gain tuning. The symmetric flight assumption (due to the bank-to-level requirement), although practical, restricts applicability in highly asymmetric or failure scenarios. Additionally, the fixed terrain scan pattern could be enhanced with adaptive or predictive scanning for improved situational awareness. Furthermore, performance over real geographic terrain requires further investigation, as the current assessment has been conducted using a simple synthetic terrain. 

In terms of robustness, the CBF-based formulation offers inherent resilience to bounded model uncertainties. However, the method still relies on nominal aerodynamic models, and its performance could degrade under severe discrepancies. Future work may explore robust or uncertainty-aware CBF formulations to maintain safety guarantees across wider operational envelopes.

\section{Conclusion}
\label{conclusion}

The challenge of unnecessary and/or untimely interventions by automatic ground collision avoidance systems (Auto-GCAS), which are often perceived as nuisances by pilots, presents a significant issue in authority-sharing between the pilot and the automated system. This study addresses this challenge by introducing a novel framework designed to ensure nuisance-free maneuvers that effectively prevent ground collisions, while eliminating the reliance on computationally intensive trajectory prediction algorithms. For this purpose, exponential control barrier functions (ECBFs) are designed in conjunction with adaptive sliding manifolds. ECBF is responsible for protecting the aircraft from a ground collision, whereas adaptive sliding manifolds ensure that interventions are timely and aggressive, which are scheduled across a range of flight states, including bank angle, pitch angle, and true airspeed. The design of adaptive sliding manifolds are formulated through an offline optimization process, wherein the objective function is constructed to align with the nuisance-free intervention criteria. Afterwards, the optimization is performed across various conditions, including pitch angle, bank angle, and true airspeed, so that the designed sliding manifolds adapt to changing flight circumstances. Additionally, flight envelope protection algorithms are designed using CBF for angle of attack, load factor, and bank angle to provide an additional layer of safety assessment for Auto-GCAS commands. This ensures that the Auto-GCAS commands remain within safe operational limits while maintaining their effectiveness in collision avoidance.

The efficacy of the overall proposed framework has been evaluated through three distinct scenarios and Monte Carlo simulations, conducted under $850$ randomized different initial conditions. In terrain collision avoidance assessments, timely and aggressive commands are generated by the Auto-GCAS, providing maximum amplitude pitch rate commands and maintaining sufficiently close proximity to the buffer zone. Additionally, an authority-sharing scenario is examined, with different pilot commands tested at various points in time. Pilot input is allowed by the Auto-GCAS as long as the imposed constraints are not violated, ensuring that unnecessary interventions do not occur. Finally, Monte Carlo simulations, including $850$ different cases, demonstrate that the proposed framework is highly effective in protecting the aircraft from ground collisions while generating timely and aggressive commands. The success is indicated by the maximum amplitude of commands, load factor, and Euclidean distance between the generated commands and the reference nuisance-free command. The proposed system achieved a remarkable success rate of $99.88\%$ across $850$ randomized Monte Carlo simulations. These results highlight the potential of the proposed framework to significantly mitigate controlled flight into terrain incidents, fostering trust and collaboration between pilots and automation systems.

Consequently, this work contributes to the broader body of safety-critical flight control by advancing the application of adaptive and constraint-aware CBF-based methods to the Auto-GCAS domain. The results offer actionable insights for researchers aiming to develop pilot-compatible, certifiable avoidance systems. Future work will focus on refining the optimization of the ECBF control gain vector in relation to dynamic terrain scan patterns. Additionally, implementing and testing the framework in real-time on a fixed-wing UAV test-bed will further validate its operational feasibility and scalability.

\section*{Acknowledgments}

During the preparation of this work the author(s) used ChatGPT in order to improve language and readability. After using this tool/service, the author(s) reviewed and edited the content as needed and take(s) full responsibility for the content of the publication. Also, the authors are deeply grateful to Dr. Ersin Daş for his valuable discussions during the presentation of this study.

\bibliography{sample}

\end{document}